\documentclass[3p, times]{elsarticle}
\makeatletter
\def\ps@pprintTitle{%
  \let\@oddhead\@empty
  \let\@evenhead\@empty
  \let\@oddfoot\@empty
  \let\@evenfoot\@oddfoot
}
\usepackage{tikz}
\usepackage{tikz}
\usetikzlibrary{shapes.geometric, arrows}

\tikzstyle{startstop} = [rectangle, rounded corners, 
text centered, 
draw=black]

\tikzstyle{io} = [trapezium, 
trapezium stretches=true, 
trapezium left angle=70, 
trapezium right angle=110,text centered, 
draw=black,text width=3cm,]

\tikzstyle{process} = [rectangle,  
draw=black, text centered]

\tikzstyle{decision} = [diamond, 
draw=black,text width=2cm, text centered]
\tikzstyle{arrow} = [thick,->,>=stealth]

\usepackage{natbib}
\usepackage{cuted}
\usepackage{flushend}
\usepackage{lipsum}
\newenvironment{proof}{\paragraph{Proof:}}{\hfill$\square$}
\usepackage{memhfixc}
\usepackage{graphicx}
\usepackage{lineno}
\usepackage{amsmath}
\usepackage{latexsym}
\usepackage{amsfonts}
\usepackage{amssymb}
\usepackage{mathrsfs}
\usepackage[utf8]{inputenc}
\usepackage{bm}
\usepackage{graphicx}
\usepackage{subfigure}

\newtheorem{theorem}{Theorem}[section]
\newtheorem{result}{Result}[section]

\newtheorem{definition}{Definition}[section]
\newtheorem{lemma}{Lemma}[section]

\newtheorem{corollary}[theorem]{Corollary}

\newcommand{\ba}{\begin{array}}
\newcommand{\ea}{\end{array}}

\newcommand{\bfl}{\begin{flushleft}}
\newcommand{\efl}{\end{flushleft}}
\newcommand{\bfr}{\begin{flushright}}
\newcommand{\efr}{\end{flushright}}
\newcommand{\bt}{\begin{theorem}}
\newcommand{\bd}{\begin{definition}}
\newcommand{\ed}{\end{definition}}
\newcommand{\et}{\end{theorem}}
\newcommand{\bl}{\begin{lemma}}
\newcommand{\el}{\end{lemma}}
\newcommand{\ee}{\end{exam}}
\newcommand{\bcor}{\begin{corollary}}
\newcommand{\ecor}{\end{corollary}}

\usepackage{multirow}
\usepackage{tabularx}
\usepackage{makecell}
\usepackage{lscape}
\usepackage{float}
\usepackage{lipsum}
\usepackage{mathtools}
\usepackage{cuted}
\usepackage{algorithm2e}

\usepackage{amssymb}

\begin{document}
\begin{frontmatter}

\title{Bayesian reliability acceptance sampling plan with optional warranty under hybrid censoring} 






\author[label1]{Rathin Das \corref{cor1}}
\address[label1]{ SQC and OR Unit, Indian Statistical Institute, Kolkata, India}

\cortext[cor1]{Corresponding author}

\ead{rathindas65@gmail.com}

\author[label1]{Biswabrata Pradhan}

\begin{abstract}
This work considers design of Bayesian reliability acceptance sampling plan (RASP) under hybrid censored life test for the products sold under optional warranty. The consumer and manufacturer agree on a common lifetime distribution of the product. However, they differ in the assessment of the prior distributions because of the adversarial nature of the consumer and manufacturer. The consumer takes decision based on his/her utility and prior belief without warranty offer by the manufacturer. If the decision is rejection, manufacturer provides warranty offer to the consumer. If the consumer rejects the lot with a warranty, the manufacturer conduct life test under hybrid censoring scheme (HCS) and provide lifetime information to the consumer. The consumer updates his/her belief based on life time information provided by the manufacturer.  The consumer then takes decision of acceptance or rejection of lot based on updated belief. The manufacturer’s task is to determine the optimal life testing plan. 

\end{abstract}


\end{frontmatter}

\section{Introduction}
The reliability of any product has a great impact on a consumer's decision. The consumer can take decision to accept or reject a lot based on the existing information about the product. In the event of rejection of a lot by the consumer, the manufacturer offer warranty to the consumer. When a lot is sold under a warranty policy, it increases the chance of accepting the lot. Due to warranty, the manufacturer incurs an expense if the product fails  during warranty period. To compensate expenses due to warranty offer, the manufacturer can increase the selling price. Nevertheless, the consumer always prefer products with low prices. Therefore, a negotiation occurs between the manufacturer and the consumer and an optional warranty is worth considering for decision-making. For non-repairable products, the rebate warranty policy is commonly used. In the rebate warranty policy, if an item fails during the warranty period, the manufacturer gives full or some proportional compensation to the consumer. Therefore, the acceptance utility of the consumer would be increased by a warranty policy. If a product from the accepted lot fails during the warranty period, then the manufacturer gives full compensation to the consumer. This warranty is called a free-replacement rebate warranty (rebate FRW) \cite{murthy1992product}. In pro-rata warranty policy (PRW) \cite{menke1969determination}, the manufacturer gives a pro-rata compensation or rebate coupon to the consumer when the product fails during the warranty period.  Here we consider a combined FRW-PRW rebate warranty policy which is a combination of rebate FRW and PRW policy that was introduced by Thomas \cite{thomas1983optimum}. The FRW policy is used in period $[0,w_1)$ and the PRW policy is used in period $[w_1,w_2)$ where $w_1<w_2$ are pre-specified time points. Note that,  it reduces to FRW policy when $w_1=w_2$ and reduces to PRW policy when $w_1=0$.

If the consumer rejects the lot, the manufacturer conducts a life test to obtain additional information and update the consumer´s belief. Lindley and Singpurwalla  \cite{lindley1991evidence,lindley1993adversarial} termed this test as adversarial life test. They considered a sequential complete life test using exponential distribution  \cite{lindley1993adversarial}. Rufo et al. \cite{rufo2014adversarial} extended it to the exponential family of distributions. Also, a Bayesian sequential negotiation model was discussed by Rufo et al. \cite{rufo2016bayesian} for multiple parties. 

 Lindley and Singpurwalla \cite{lindley1993adversarial} obtained RASP based on complete data. However, due to time, cost, and other resource limitations, censored life tests are conducted to collect lifetime information. Type-I, type-II and hybrid censoring schemes are the most common types of censoring schemes applied in life-testing. In type-I censoring, the life test is terminated at a predetermined time $T_0$ and in type-II censoring scheme, the life test is terminated after a fixed number of failures $r<n$. The HCS is a combination of type-I and type-II censoring schemes. Two types of HCSs are considered in practice, namely type-I HCS and type-II HCS. In type-I HCS, the life test is terminated after the failure of $r^{th}$ or at the time $T_0$, whichever is earlier (see  \cite{epstein1954truncated}).

There have been a number of works on determination Bayesian variable sampling plans under different censoring schemes for the exponential distributions. For example, Yeh \cite{yeh1990optimal,yeh1994bayesian} considered the optimal design of RASPs using the Bayesian decision-theoretic approach under type II censoring and type I censoring schemes. Yeh and Choy \cite{yeh1995bayesian} considered it under random censoring scheme. Chen et al. \cite{chen2007bayesian} and Lin et al. \cite{lin2008exact,lin2010corrections} considered determination of RASP under HCS. Lin \cite{lin2002bayesian} introduced the Bayes decision function for the exponential distribution under type-I censoring. For hybrid censored data, the Bayes decision was introduced by Liang and Yang \cite{liang2013optimal}. These works considered the minimization of total expected costs with respect to a single prior distribution for the parameters of the lifetime distribution and utility function agreed by both manufacturer and consumer. We consider that the consumer and the manufacturer agree on a common lifetime distribution but differ on prior distributions and 
 utility functions because of their adversarial nature. In this context, there is no work on determination of RASP under censoring, all works were based on a single decision criterion and the warranty of a product was not used in the negotiation. In our work, we conduct the life test under HCS. The warranty of the product is used in the negotiation. 

In our work, first, we consider that the manufacturer not only knows his/her own utility and prior but also knows the consumer’s utility, prior and actions. However, this is a strong knowledge assumption by the manufacturer. In real-life scenarios, sometimes the manufacturer does not know the consumer’s utility, prior
and actions. We introduce a sampling plan called random decision-making sampling plan (RDSP) where consumer's action is uncertain. It is assumed that the form of the consumer’s utility function and prior distribution are known to the manufacturer, but the parameters of utility function
and prior are unknown to the manufacturer. It is assumed that these parameters are random. 

The organization of the paper is as follows. The model is discussed in Section \ref{model}. Determination of optimum RASPs for exponential and Weibull distributions are considered in Section \ref{we}. The RDSP for exponential distribution is discussed in section \ref{ARA}. Numerical example in different situations is discussed in Section \ref{numerical}. A real data set is analyzed to demonstrate the proposed model in Section \ref{real}.  The conclusion is made in Section \ref{conclu}.
\section{Development of the Model with an optional warranty}\label{model}
Suppose $X$ denotes the lifetime of a product with cumulative distribution function (cdf) $F_{\boldsymbol{\theta}}$ and probability density function (pdf) $f_{\boldsymbol{\theta}}$, where $\boldsymbol{\theta}=(\theta_1,\cdots,\theta_p)$ is a vector of parameters. The manufacturer negotiates with the consumer to sale of his/her product. The consumer may accept or reject the
lot based on prior $p_C(\boldsymbol{\theta})$ of $\boldsymbol{\theta}$ and a utility function $\mathcal{U}_C(. \ | \ \boldsymbol{\theta})$, $.\equiv \mathcal{A}, \mathcal{R}$, where $\mathcal{A}$ and $\mathcal{R}$ denote acceptance and rejection, respectively. The consumer accepts the lot  (Lindley and Singpurwalla \cite{lindley1991evidence}) if
 \begin{align}\label{1} E_{p_C(\boldsymbol{\theta})}(\mathcal{U}_C(\mathcal{A}\ | \ \boldsymbol{\theta})) 
 \geq E_{p_C(\boldsymbol{\theta})}(\mathcal{U}_C(\mathcal{R} \ | \ \boldsymbol{\theta})). 
 \end{align}
 If $\mathcal{U}_C(\mathcal{A}\ | \ \boldsymbol{\theta})$ and $\mathcal{U}_C(\mathcal{A}\ | \ \boldsymbol{\theta})$ are taken as a loss functions instead of utility (profit) functions, the inequality is reversed.
 Initially, the manufacturer offers a lot to the consumer. The decision of acceptance or rejection of lot is taken by the consumer based on initial utility function $\mathcal{U}^0_C(. \ | \ \boldsymbol{\theta})$.
If the utility function $\mathcal{U}^0_{\mathcal{C}}(. \ |\  \boldsymbol{\theta})$ fail to satisfy inequality (\ref{1}), the manufacturer offers a warranty to the consumer, which can be purchased with the product at an additional cost. The consumer changes his/her utility function with a warranty offer from the manufacturer.  Now, the consumer takes decision with the warranty using the inequality (\ref{1}). If the inequality is satisfied, the lot is accepted with the warranty offer. If the inequality (\ref{1}) is not satisfied with warranty offer, the manufacturer performs a life test to update the consumer's belief. Finally, the decision is taken by the consumer using updated utility function based on lifetime data.

\subsection{Consumer's utility function}
The consumer has a requirement of minimum lifetime $L$
for the product. If the product fails before the lifetime $L$, the consumer has a certain amount of loss. It is assumed that the loss function is linearly decreasing with his/her lifetime. Therefore the loss function  in terms of lifetime can be taken as
\begin{align}\label{theta}
    \mathcal{U}^0_\mathcal{C}(\mathcal{A}\ |\ X)=\begin{cases}
       a_1\left(1-\frac{X}{L}\right)+a_2&\text{if }X\leq L\\
       a_2&\text{if }X> L
    \end{cases}
\end{align}
and 
\begin{align}\label{theta1}
    \mathcal{U}^0_\mathcal{C}(\mathcal{R}\ |\ X)=a_3
\end{align}
where $a_1$ is the proportional loss with the lifetime of the product if it fails before the time $L$, $a_2$ is the fixed cost which is paid by the consumer when the lot is accepted and $a_3$ is the loss due to rejection. The
consumer accepts the lot without warranty and life-testing if
\begin{align}\label{21}
     E_{p_C(\boldsymbol{\theta})}\left[E_{X\ |\ \boldsymbol{\theta}}\left\{\mathcal{U}^0_C(\mathcal{A}\ |\ X)\right\}\right]\leq a_3.
\end{align}
If the inequality (\ref{21}) does not hold, the manufacturer offers a warranty to the consumer. The consumer updates his/her utility function using warranty offer. Here, the manufacturer sells his/her products under the combined FRW-PRW rebate warranty policy, as described in Introduction.  Let $c_s$ be the selling price of the item without any warranty. If the lot is sold with a warranty, the consumer needs to pay the additional cost $c_w$ per item for the lot due to warranty offer. Therefore, $(c_s+c_w)$ is the selling price of an item with the warranty. Let $q(X)$ be the amount of rebate of an item with lifetime $X$. Then, from Murthy and Blischke \cite{murthy1992product} for rebate warranty, the expected cost for an item to the consumer is $(c_s+c_w)-E_{X\ |\ \boldsymbol{\theta}}\left[q(X)\right]$ and to the manufacturer it is $c_m+E_{X\ |\ \boldsymbol{\theta}}\left[q(X)\right]$, where $c_m$ is the cost of supplying an item.
 The cost $q(X)$ under combined FRW-PRW rebate warranty policy (see \cite{thomas1983optimum}) is given by
\begin{align}
q(X)=
\begin{cases}
    (c_s+c_w)&\text{if }0\leq X<w_1\\
    (c_s+c_w)\frac{w_1-X}{w_1-w_2}&\text{if }w_1\leq X<w_2\\
    0 &\text{if }X>w_2.
    \end{cases}
\end{align}
So the expected cost of accepting the lot with warranty is 
\begin{align}\label{theta2}
    E_{X\ |\ \boldsymbol{\theta}}\left\{\mathcal{U}^0_C(\mathcal{A}\ |\ X)\right\}-E_{X\ |\ \boldsymbol{\theta}}\left\{q(X)\right\}+c_w. 
\end{align}
Therefore, the consumer accepts the lot with a warranty and without life-testing if
\begin{align}\label{22}
 & E_{p_C(\boldsymbol{\theta})}\left[E_{X\ |\ \boldsymbol{\theta}}\left\{\mathcal{U}^0_C(\mathcal{A}\ |\ X)\right\}\right]-E_{p_C(\boldsymbol{\theta})}\left[E_{X\ |\ \boldsymbol{\theta}}\left\{q(X)\right\}\right]+c_w\leq a_3.
\end{align} 
We have
\begin{align*}
    E_{X\ |\ \boldsymbol{\theta}}[\mathcal{U}^0_{\mathcal{C}}(\mathcal{A}\ |\ X)]=\frac{a_1}{L}\int_{0}^{L}F_{\boldsymbol{\theta}}(x) dx +a_2.
\end{align*}
From Budhiraja \& Pradhan \cite{budhiraja2019optimum} we get,
\begin{align*}
    E_{X\ |\ \boldsymbol{\theta}}[q(X)]=\frac{c_s+c_w}{w_2-w_1}\int_{w_1}^{w_2}F_{\boldsymbol{\theta}}(x) dx.
\end{align*}
If the inequality (\ref{22}) does not hold, the consumer's decision is rejection with a warranty offer. In the event of rejection of lot with warranty offer, the manufacturer performs a life test. The consumer updates his/her belief based on lifetime information obtained from the life test. Now based on the updated belief, the consumer takes the decision of rejection or acceptance of the lot.  

\subsection{Derivation of consumer's decision}
Suppose $n$ items are put on a life test under Type-I HCS under the design $\boldsymbol{m}=(n,r,T_0)$. The lifetimes of $n$ the items $X_1,\cdots,X_n$ are independent and identically distributed (iid) with common cdf $F_{\boldsymbol{\theta}}$ and pdf $f_{\boldsymbol{\theta}}$. Let $X_{(1)}\leq X_{(2)}\leq \cdots \leq X_{(n)}$ be the order failure times of $n$ items. Let $D$ and $\eta$ be the number of failures and duration of the test, respectively.  The observed data under Type-I HCS is represented by
  \begin{align*}
      \boldsymbol{x}=\begin{cases}
         (d=0)& \text{if } x_{(1)}>T_0 \\(x_{(1)},x_{(2)},\cdots,x_{(r)},r)& \text{if } x_{(r)}<T_0\\
          (x_{(1)},x_{(2)},\cdots,x_{(d)},d) & \text{if } x_{(d)}<T_0<x_{(d+1)}<\cdots<x_{(r)}, d<r,
      \end{cases}
  \end{align*} 
where $x_{(i)}$ and $d$ are the observed values of $X_{(i)}$ and $D$, respectively.
The likelihood function under design $\boldsymbol{m}=(n,r,T_0)$ is given by:
\begin{align*}
    L(\boldsymbol{\theta}\ |\ \boldsymbol{x},\boldsymbol{m})=\begin{cases}
   \left( F_{\boldsymbol{\theta}}(T_0)\right)^{n}&\text{if }d=0\\
    \frac{n!}{(n-d)!}\prod_{i=0}^df_{\boldsymbol{\theta}}(x_{(i)})(1-F_{\boldsymbol{\theta}}(T_0))^{n-d}&\text{if }d=1,\cdots,r-1\\
    \frac{n!}{(n-r)!}\prod_{i=1}^df_{\boldsymbol{\theta}}(x_{(i)})(1-F_{\boldsymbol{\theta}}(x_{(r)}))^{n-r}&\text{if }d=r\\
    \end{cases}
\end{align*} 
Therefore, the posterior distribution of $\boldsymbol{\theta}$ given data $\boldsymbol{x}$ under the design $\boldsymbol{m}=(n,r,T_0)$ is given by
\begin{equation}\label{pe}
p_\mathcal{C}(\boldsymbol{\theta}\ | \ \boldsymbol{x},\boldsymbol{m})= \frac{L(\boldsymbol{\theta}\ | \ \boldsymbol{x},\boldsymbol{m})p_C(\boldsymbol{\theta})}{\int_{\boldsymbol{\theta}} L(\boldsymbol{\theta}\ |\ \boldsymbol{x},\boldsymbol{m})p_C(\boldsymbol{\theta})d\boldsymbol{\theta}} .   
\end{equation}

The consumer accepts the lot after life testing and  without warranty based on the observed data $\boldsymbol{x}$ if
\begin{align}\label{in1}
    \mathcal{U}^0_C(\mathcal{A}\ |\ \boldsymbol{x},\boldsymbol{m})&=\int_\theta \mathcal{U}^0_{\mathcal{C}}(\mathcal{A}\ |\ \boldsymbol{\theta})p_{\mathcal{C}}(\boldsymbol{\theta}\ | \ \boldsymbol{x},\boldsymbol{m})d\boldsymbol{\theta}\nonumber\\
    &\leq \int_\theta \mathcal{U}^0_{\mathcal{C}}(\mathcal{R}\ |\ \boldsymbol{\theta})p_{\mathcal{C}}(\boldsymbol{\theta}\ |\ \boldsymbol{x},\boldsymbol{m})d\boldsymbol{\theta}\nonumber\\
   & =\mathcal{U}^0_{\mathcal{C}}(\mathcal{R}\ |\ \boldsymbol{x},\boldsymbol{m})=a_3.
\end{align}
If the inequality (\ref{in1}) is not satisfied, the manufacturer offers a warranty to the consumer. Now, the consumer will accept the lot after life testing with warranty if
\begin{align}\label{in2}
    \mathcal{U}^0_\mathcal{C}(\mathcal{A}\ |\ \boldsymbol{x},\boldsymbol{m})-E_{\boldsymbol{\theta}\ | \ (\boldsymbol{x},\boldsymbol{m})}\left(E_{X\ |\ \boldsymbol{\theta}}[q(X)]\right)+c_w\leq a_3.
\end{align}
If the inequality is not satisfied after life-testing with a warranty offer, the manufacturer fails to convince the consumer and the consumer rejects the lot.

We now define three sets of data for $\boldsymbol{x}$ values as follows: $\mathcal{X}$ is the set of all possible values of $\boldsymbol{x}$ which satisfies the inequality (\ref{in1}), $\mathcal{Y}$ is the set of all possible values of $\boldsymbol{x}$ which does not satisfy the inequality (\ref{in1}) and satisfies the inequality (\ref{in2}) and $\mathcal{Z}$ is the set of all possible values of $\boldsymbol{x}$ which do not satisfy both inequalities (\ref{in1}) and (\ref{in2}). Therefore $\mathcal{X}$ and $\mathcal{Y}$ contain the values of $\boldsymbol{x}$ for which the consumer accepts the lot without warranty and with warranty, respectively and $\mathcal{Z}$ contains the values of  $\boldsymbol{x}$ for which  consumer reject the lot. The flowchart of the consumer's decision at every steps is given in Figure \ref{f1}.
\begin{figure}[hbt!]

\caption{\large{The procedure of consumer's actions}~~~~~~~~~~~~~~~~~~~~~~~~~~~~~~~}\label{f1}
\vspace{0.1cm}
\begin{center}
\begin{tikzpicture}[node distance=1cm]
\node(s1)[startstop]{Consumer comes to the manufacturer with his/her belief};
    \node (in1) [process,below of=s1,yshift=-0.2cm] {Calculate $E_{p_C(\boldsymbol{\theta})}\left[E_{X\ |\ \boldsymbol{\theta}}\left\{\mathcal{U}^0_C(\mathcal{A}\ |\ X)\right\}\right]$};
    \node (dec1) [decision, below of=in1,yshift=-1.4cm] {Check the inequality (\ref{21})};
    \node (pro2b) [startstop, right of=dec1, xshift=4.2cm,text width=5cm] {Consumer accepts the lot without warranty and without life testing};
    \node (dec3) [process, below of=dec1,yshift=-1.4cm] {Calculate $E_{p_C(\boldsymbol{\theta})}\left[E_{X\ |\ \boldsymbol{\theta}}\left\{q(X)\right\}\right]$};
    \node (dec2) [decision, below of=dec3,yshift=-1.4cm] {Check the inequality (\ref{22})};
    \node (pro6b) [startstop, right of=dec2, xshift=4.2cm,text width=5cm] {Consumer accepts the lot with warranty and without life testing};
    \node (p1) [process, below of=dec2,yshift=-1.3cm] {Start life-testing to update consumer's belief};
    \draw[arrow](s1) --(in1);
    \node (dec4) [decision, below of=p1,yshift=-1.3cm] {Check the inequality (\ref{in1})};
    \node (pro6) [startstop, right of=dec4, xshift=4.2cm,text width=5cm] {Consumer accepts the lot without warranty and with life testing};
    \node (dec5) [decision, below of=dec4,yshift=-2.5cm] {Check the inequality (\ref{in2})};
    \node (pro61) [startstop, right of=dec5, xshift=4.2cm,text width=5cm] {Consumer accepts the lot with warranty and with life testing};
    \node (pro62) [startstop, below of=dec5,yshift=-1.5cm,text width=5cm] {Consumer rejects the lot };
\draw [arrow] (in1) -- (dec1);
\draw [arrow] (dec1) -- node[anchor=north] {yes} (pro2b);
\draw [arrow] (dec1) -- node[anchor=east] {no} (dec3);
\draw [arrow] (dec2) -- node[anchor=north] {yes} (pro6b);
\draw [arrow] (dec3) -- (dec2);
\draw [arrow] (dec2) -- node[anchor=east] {no} (p1);
\draw [arrow] (p1) -- (dec4);
\draw [arrow] (dec4) -- node[anchor=east] {no} (dec5);
\draw [arrow] (dec4) -- node[anchor=north] {yes} (pro6);
\draw [arrow] (dec5) -- node[anchor=east] {no} (pro62);
\draw [arrow] (dec5) -- node[anchor=north] {yes} (pro61);
\end{tikzpicture}
\end{center}
\end{figure}
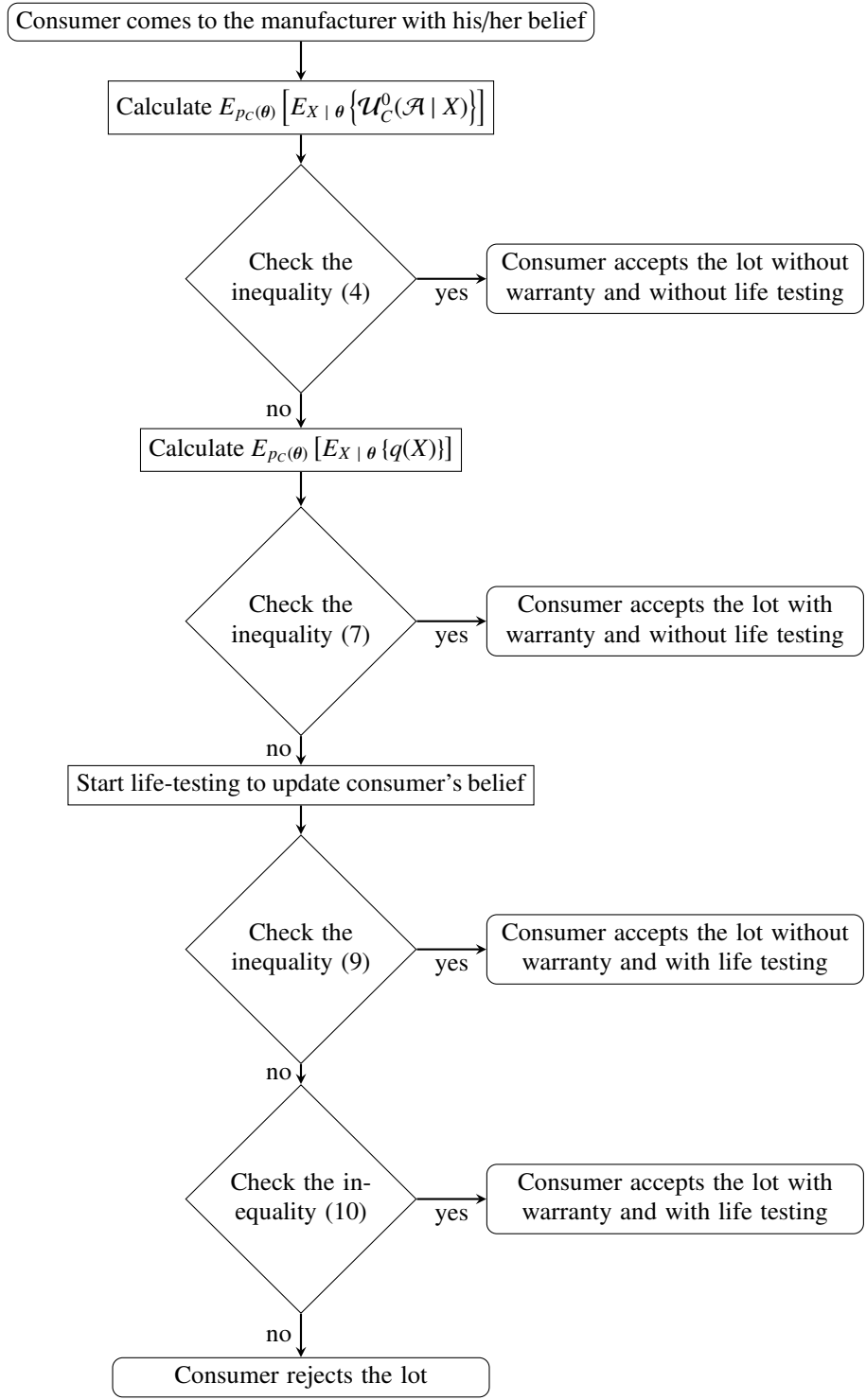

Next, we consider the manufacturer's decision of optimal planning of the life test. The manufacturer needs to determine optimum values of $\boldsymbol{m}=(n,r,T_0)$. The decision is taken based on the manufacturer's utility function
\subsection{Manufacturer's utility function}
We consider the utility functions of the manufacturer corresponding to accepting the lot without warranty and with a warranty after life testing for a given data $\boldsymbol{x}$ are
\begin{align}\label{acc}
   \mathcal{U}_\mathcal{M}(\mathcal{A}_{wo}\ |\ \boldsymbol{x}, \boldsymbol{\theta})=b_1E_{X\ | \ (\boldsymbol{x }, \boldsymbol{\theta})}[X^q]-b_2-b_4n-b_5d-b_6\eta,
   \end{align}
   and
   \begin{align}\label{acc1}
   \mathcal{U}_\mathcal{M}(\mathcal{A}_{w}\ |\ \boldsymbol{x}, \boldsymbol{\theta})=b_1E_{X\ | \ (\boldsymbol{x }, \boldsymbol{\theta})}[X^q]+c_w-b_2-E_{X\ | \ (\boldsymbol{x }, \boldsymbol{\theta})}[q(X)]-b_4n-b_5d-b_6\eta,
   \end{align}
respectively, where $b_1>0$ represents the profit to the manufacturer per unit time that item functions, $q$ is the risk parameter (see \cite{happich2001utility}), and $b_2$ represents the fixed cost when the lot is accepted. The utility corresponding to rejection is taken as 
   \begin{align}\label{rej}
    \mathcal{U}_\mathcal{M}(\mathcal{R}\ |\ \boldsymbol{x}, \boldsymbol{\theta})=b_3-b_4n-b_5d-b_6\eta,
\end{align}
where $b_3$ is the cost of rejecting the lot, $b_4$ is the cost of putting an item on test, $b_5$ is the cost per failed item, $b_6$ is the cost of running the test per unit time,
 $\eta$ is the duration of the test
and $D$ is the number of failures up to time $\eta$. Note that
\begin{align*}
    E_{X\ | \ (\boldsymbol{x }, \boldsymbol{\theta})}[X^q]=\int_0^\infty x^q f_{(\boldsymbol{x}, \boldsymbol{\theta})}(x)dx= \int_0^\infty x^q f_{\boldsymbol{\theta}}(x)dx
\end{align*}

 Using the consumer's decision after life-testing, we find the manufacturer's expected utility to find the optimum life-testing plan.

\subsection{Derivation of manufacturer's Decision} \label{deci}
The manufacturer's utility function for the given data $\boldsymbol{x}$ is given by
    \begin{align*}
    \psi_\mathcal{M}(\boldsymbol{m}\ |\  \boldsymbol{x},\boldsymbol{\theta})=&I_{\mathcal{X}}(\boldsymbol{x})\left[b_1E_{X\ |\ (\boldsymbol{x},\boldsymbol{\theta})}(X^q)-b_2-b_4n-b_5d-b_6\eta\right]+I_{\mathcal{Y}}(\boldsymbol{x})\left[b_1E_{X|(\boldsymbol{x},\boldsymbol{\theta})}(X^q)-b_2-b_4n-b_5d-b_6\eta\right]\\
    &+I_{\mathcal{Z}}(\boldsymbol{x})\left[b_3-b_4n-b_5d-b_6\eta\right],
    \end{align*}
    where $I_{\mathcal{S}}(\boldsymbol{s})$ denotes the indicator function which is defined as
    \begin{align*}
    I_{\mathcal{S}}(\boldsymbol{s})=
        \begin{cases}
           1 & s\in \mathcal{S}\\
           0 & s\notin \mathcal{S}.
        \end{cases}
    \end{align*}
     The expected utility with respect to data $\boldsymbol{x}$ and $\boldsymbol{\theta}$ is given by
    \begin{align}\label{manu1}
    \psi_\mathcal{M}(\boldsymbol{m})
    =&E_{\boldsymbol{\theta}}[E_{\boldsymbol{x}\ |\ \boldsymbol{\theta}}[\psi_\mathcal{M}(\boldsymbol{m}\ |\  \boldsymbol{x},\boldsymbol{\theta})]]\nonumber\\
    =&\int_{\boldsymbol{\theta}}\int_{\boldsymbol{y}}\sum_{d}\left[I_{\mathcal{X}\cup\mathcal{Y}}(\boldsymbol{x})\left\{b_1E_{X|(\boldsymbol{x},\boldsymbol{\theta})}(X^q)-b_2\right\}-I_{\mathcal{Y}}(\boldsymbol{x})E_{X|(\boldsymbol{x},\boldsymbol{\theta})}
    (q(X))+I_{\mathcal{Z}}(\boldsymbol{x})b_3\right]f(\boldsymbol{y},d\ | \ \boldsymbol{\theta})\nonumber\\
    &-b_4n-b_5E[D\ |\ \boldsymbol{\theta}]-b_6E[\eta\ |\ \boldsymbol{\theta}]p_\mathcal{M}(\boldsymbol{\theta})~d\boldsymbol{y}~d\boldsymbol{\theta},
    \end{align}
    where $\boldsymbol{x}=(\boldsymbol{y},d)$, $\boldsymbol{y}=(x_{(1)},\ldots,x_{(d)})$, $f(\boldsymbol{y},d\ | \ \boldsymbol{\theta})$ be the distribution of $\boldsymbol{x}$ and $p_\mathcal{M}(\boldsymbol{\theta})$ is the prior distribution of $\boldsymbol{\theta}$ for the manufacturer.\\
Let $A_{wo}$ and $A_w$ be the events that the consumer accepts the lot without warranty and with warranty, respectively, and $R$ be the event that the consumer rejects the lot.
    Then,
    \begin{align}
P(A_{wo}\ | \ \theta)&=\int_{\boldsymbol{y}}\sum_{d}I_{\mathcal{X}}(x)f(\boldsymbol{y},d\ | \ \boldsymbol{\theta})~d\boldsymbol{y},
    \end{align}
    \begin{align}
P(A_{w}\ | \ \theta)&=\int_{\boldsymbol{y}}\sum_{d}I_{\mathcal{Y}}(x)f(\boldsymbol{y},d\ | \ \boldsymbol{\theta})~d\boldsymbol{y}
    \end{align}
    and
    \begin{align}
P(R\ | \ \boldsymbol{\theta})&=\int_{\boldsymbol{y}}\sum_{d}I_{\mathcal{Z}}(x)f(\boldsymbol{y},d\ | \ \boldsymbol{\theta})~d\boldsymbol{y}.
    \end{align}
Let $L_w$ be the expected loss of the manufacturer due to the warranty policy. Then $L_w(\boldsymbol{\theta})$ can be written as
 \begin{align}
     L_w(\boldsymbol{\theta})=\int_{\boldsymbol{y}}\sum_{d}\left[E_{X\ |\ (\boldsymbol{x},\boldsymbol{\theta})}(q(X))-c_w\right]I_{\mathcal{Y}}(x)f(\boldsymbol{y},d\ | \ \boldsymbol{\theta})~d\boldsymbol{y}.
 \end{align}
 The expected number of failures is given by (see
 \cite{bhattacharya2014optimum})
\begin{align*}
    E[D\ | \ \boldsymbol{\theta}]=\sum_{j=0}^{r-1}j\binom{n}{j}(1-F_{\boldsymbol{\theta}}(T_0))^j(F_{\boldsymbol{\theta}}(T_0))^{n-j}+r\sum_{j=r}^n\binom{n}{j}(1-F_{\boldsymbol{\theta}}(T_0))^j(F_{\boldsymbol{\theta}}(T_0))^{n-j}.
\end{align*}
The expected duration of the test is given by (see
 \cite{bhattacharya2014optimum})
\begin{align*}
    E(\eta\ | \ \boldsymbol{\theta})=&T_0\left(1-\sum_{j=r}^n\binom{n}{j}(F_{\boldsymbol{\theta}}(T_0)^j(1-F_{\boldsymbol{\theta}}(T_0))^{n-j}\right)+r\binom{n}{r}\int_0^{T_0}x(F_{\boldsymbol{\theta}}(x))^{r-1}(1-F_{\boldsymbol{\theta}}(x))^{n-r}f_{\boldsymbol{\theta}}(x)dx.
\end{align*}
 The expected utility of the manufacturer can be written as
 \begin{align}
     \psi_\mathcal{M}(\boldsymbol{m})
=&\int_{\boldsymbol{\theta}}\sum_{\boldsymbol{x}}\left[b_1E_{X\ |\ (\boldsymbol{x},\boldsymbol{\theta})}(X^q)[P(A_{wo}\ |\ \boldsymbol{\theta})+P(A_w\ |\ \boldsymbol{\theta})\right]p_M(\boldsymbol{\theta})d\boldsymbol{\theta}\nonumber\\
&-b_2[P(A_{wo})+P(A_w)]-L_w+b_3P(R)-b_4n-b_5E(D)-b_6E(\eta),
 \end{align}
\resizebox{\textwidth}{!}{ where $P({A}_{wo})=\int_{\boldsymbol{\theta}}P({A}_{wo}\ |\ \boldsymbol{\theta})p_\mathcal{M}(\boldsymbol{\theta})d\boldsymbol{\theta}$, $P({A}_{ww})=\int_{\boldsymbol{\theta}}P({A}_{ww}\ |\ \boldsymbol{\theta})p_\mathcal{M}(\boldsymbol{\theta})d\boldsymbol{\theta}$, $P({R})=\int_{\boldsymbol{\theta}}P({R}\ |\ \boldsymbol{\theta})p_\mathcal{M}(\boldsymbol{\theta})d\boldsymbol{\theta}$, $L_w=\int_{\boldsymbol{\theta}}L_w(\boldsymbol{\theta})p_\mathcal{M}(\boldsymbol{\theta})d\boldsymbol{\theta}$ } $E[{D}]=\int_{\boldsymbol{\theta}}E({D}\ |\ \boldsymbol{\theta})p_\mathcal{M}(\boldsymbol{\theta})d\boldsymbol{\theta}$ and $E[{\eta}]=\int_{\boldsymbol{\theta}}E({\eta}\ |\ \boldsymbol{\theta})p_\mathcal{M}(\boldsymbol{\theta})d\boldsymbol{\theta}.$
 \vspace{0.2cm}
 
    The optimal values of $\boldsymbol{m}$ is given by
    \begin{align*} \boldsymbol{m}^*=(n^*,r^*,T_0^*)=\arg\max\limits_{\boldsymbol{m}}\psi_{\mathcal{M}}(\boldsymbol{m}).
    \end{align*}
\textbf{Note :}
\begin{enumerate}
    \item The manufacturer's utility of accepting the lot without warranty and life-testing is 
    \begin{align*}
        \psi_{\mathcal{M}}((0,0,0))=\int_{\boldsymbol{\theta}}b_1E_{X\ |\ \boldsymbol{\theta}}(X^q)p_{\mathcal{M}}(\boldsymbol{\theta})~d\boldsymbol{\theta}-b_2.
    \end{align*}
     \item The manufacturer's utility of accepting the lot with warranty and without life-testing is 
    \begin{align*}
        \psi_{\mathcal{M}}((0,0,0))=\int_{\boldsymbol{\theta}}b_1E_{X\ |\ \boldsymbol{\theta}}\{X^q-q(X)\}p_{\mathcal{M}}(\boldsymbol{\theta})~d\boldsymbol{\theta}-b_2+c_w.
        \end{align*}
         \item The manufacturer's utility of rejecting the lot without life-testing is 
    \begin{align*}
        \psi_{\mathcal{M}}((0,0,0))=b_3.
        \end{align*}
        \item After life testing the manufacturer's utility must be greater than the above three values. The conditions under which the lot is accepted without life testing are given in (\ref{in1}) and (\ref{in2}). If the manufacturer's utility is less than $b_3$ among all sampling plans, then the lot is rejected without life testing.
\end{enumerate}

\section{Optimal RASP }\label{we}
Here we obtain optimal RASP for exponential and Weibull distributions. 
\subsection{Exponential distribution case}\label{exp}
We obtain optimal RASP when lifetime distribution follows exponential distribution. The lifetimes of the items be i.i.d with pdf $f_\theta(t)=\frac{1}{\theta}e^{-\frac{t}{\theta}}$, $\theta>0$, $t>0$, denoted by $\text{Exp}(\theta)$.
\subsubsection{Consumer's decision}\label{cd}
 It is assumed that $\theta$ follows inverse gamma distribution  with pdf 
\begin{align*}
    p_C(\theta)=\frac{\beta_1^{\alpha_1}}{\Gamma(\alpha_1)}\theta^{-\alpha_1-1}\exp\left(-\frac{\beta_1}{\theta}\right), ~~\theta>0, ~\alpha_1>0,~\beta_1>0.
\end{align*}
For the Exp($\theta$), the UFs corresponding to acceptance and rejection of the consumer without warranty are given by
 $$\mathcal{U}^0_C(\mathcal{A}\ | \ \theta)=E_{X\ |\ \theta}\left[\mathcal{U}_{\mathcal{C}}^0\left(\mathcal{A}\ |\ X\right)\right]=\frac{a_1}{L}\left[L+\theta\exp(-\theta/L)-\theta\right]+a_2$$ and
  $$\mathcal{U}^0_C(\mathcal{R}\ | \ \theta)=E_{X\ |\ \theta}\left[\mathcal{U}_{\mathcal{C}}^0\left(\mathcal{R}\ |\ X\right)\right]=a_3,$$
respectively. Using (\ref{21}), the consumer accepts the lot without life-testing and without warranty if
\begin{align}\label{i1}
    \frac{a_1}{L}\left[L+\frac{\beta_1^{\alpha_1}}{(\alpha_1-1)}\left(\frac{1}{(\beta_1+L)}\right)^{\alpha_1-1}-\frac{\beta_1}{\alpha_1-1}\right]+a_2\leq a_3.
\end{align}
If the above inequality is not satisfied, the manufacturer gives a warranty offer to the consumer. The expected warranty cost is obtained as 
\begin{align*}
    E_{X\ | \ \theta}[q(X)]=\frac{cs+cw}{w_2-w_1}\left[(w_2-w_1)+\theta\left\{\exp(-w_2/\theta)-\exp(w_1/\theta)\right\}\right].
\end{align*}
Taking expectation with respect to $\theta$, we get
\begin{align*}
    E_{\theta}[E_{X\ | \ \theta}[q(X)]]=(c_s+c_w)-\frac{c_s+c_w}{w_2-w_1}\frac{\beta_1^{\alpha_1}}{(\alpha_1-1)}\left[\frac{1}{(\beta_1+w_1)^{\alpha_1-1}}-\frac{1}{(\beta_1+w_2)^{\alpha_1-1}}\right].
\end{align*}
Using (\ref{22}), the consumer accepts the lot without life-testing and with a warranty, if
\begin{align}\label{i2}
    \frac{a_1}{L}\left[L+\frac{\beta_1^{\alpha_1}}{(\alpha_1-1)}\left(\frac{1}{(\beta_1+L)}\right)^{\alpha_1-1}-\frac{\beta_1}{\alpha_1-1}\right]+a_2+\frac{c_s+c_w}{w_2-w_1}\frac{\beta_1^{\alpha_1}}{(\alpha_1-1)}\left[\frac{1}{(\beta_1+w_1)^{\alpha_1-1}}-\frac{1}{(\beta_1+w_2)^{\alpha_1-1}}\right]\geq {a_3+c_s}.
\end{align}
If both inequalities (\ref{i1}) and (\ref{i2}) are not satisfied, the manufacturer conducts a life test and provides data to the consumer. Suppose $n$ items are put on test under type-I HCS. 
The likelihood function is given by
\begin{align}\label{e1}
    L(\theta\ |\ \boldsymbol{x},\boldsymbol{m})\propto \theta^{-d}\exp\left(-\frac{v(\boldsymbol{x})}{\theta}\right),
\end{align}
 where \begin{align}
    v(\boldsymbol{x})=\begin{cases}
    nT_0& d=0\\
     \sum_{i=1}^d x_{(i)}+(n-d)T_0 &1\leq d<r\\
   \sum_{i=1}^r x_{(i)}+(n-r)x_{(r)} &d=r.
    \end{cases}
\end{align}
The upper bound of $v(\boldsymbol{x})$ is $nT_0$ for all $d$.
The posterior distribution of $\theta$ is given by
\begin{align*}
    p(\theta\ |\ \boldsymbol{x},\boldsymbol{m})\propto \theta^{-d-\alpha_1-1}\exp\left(-\frac{v(\boldsymbol{x})+\beta_1}{\theta}\right).
\end{align*}
This shows that $\theta$ follows inverse gamma distribution with parameter
$(\alpha_1+d,\beta_1+v(\boldsymbol{x}))$. Using (\ref{in1}),  we get,
\begin{align}\label{w1}
    \mathcal{U}^0_C(\mathcal{A}\ |\ \boldsymbol{x},\boldsymbol{m})&=\int_{0}^\infty \left[\frac{a_1}{L}\left[L+\theta\exp(-\theta/L)-\theta\right]+a_2 \right]\frac{(v(\boldsymbol{x})+\beta_1)^{\alpha_1+d}}{\Gamma(\alpha_1+d)}\theta^{-\alpha_1+d-1}\exp\left(-\frac{v(\boldsymbol{x})+\beta_1}{\theta}\right)d\theta\leq a_3\nonumber\\
    &\implies\frac{a_1}{L}\left[L+\frac{(\beta_1+v(\boldsymbol{x}))^{\alpha_1+d}}{(\alpha_1+d-1)}\left(\frac{1}{(\beta_1+v(\boldsymbol{x})+L)}\right)^{\alpha_1+d-1}-\frac{\beta_1+v(\boldsymbol{x})}{\alpha_1+d-1}\right]+a_2\leq a_3\nonumber\\
    &\implies A_1((v(\boldsymbol{x}),d)\ |\ \boldsymbol{m})\leq \frac{L(-a_2+a_3)}{a_1},
\end{align}
where 
\begin{align*}
  A_1((v(\boldsymbol{x}),d)\ |\ \boldsymbol{m})=   \left[L+\frac{(\beta_1+v(\boldsymbol{x}))^{\alpha_1+d}}{(\alpha_1+d-1)}\left(\frac{1}{(\beta_1+v(\boldsymbol{x})+L)}\right)^{\alpha_1+d-1}-\frac{\beta_1+v(\boldsymbol{x})}{\alpha_1+d-1}\right].
\end{align*}
If the inequality (\ref{w1}) is satisfied, the consumer accepts the lot without warranty. If the inequality (\ref{w1}) is not satisfied, the manufacturer gives a warranty offer to the consumer. Now, the consumer will accept the lot with a warranty, using (\ref{in2}), if
\begin{align}\label{w2}
   &\int_{0}^\infty\left[\mathcal{U}^0_C(\mathcal{A}\ |\ \theta)-\frac{c_s+c_w}{w_2-w_1}\int_{w_1}^{w_2}[1-\exp(-x/\theta)]dx-c_w\right]\frac{(v(\boldsymbol{x})+\beta_1)^{\alpha_1+d}}{\Gamma(\alpha_1+d)}\theta^{-\alpha_1+d-1}\exp\left(-\frac{v(\boldsymbol{x})+\beta_1}{\theta}\right)d\theta\leq a_3\nonumber\\
    &\implies \frac{a_1}{L}A_1((v(\boldsymbol{x}),d)\ |\ \boldsymbol{m})-A_2((v(\boldsymbol{x}),d)\ |\ \boldsymbol{m})\leq {-a_2+a_3},
\end{align}
where
\begin{align*}
    A_2((v(\boldsymbol{x}),d)\ |\ \boldsymbol{m})=c_s-\frac{c_s+c_w}{w_2-w_1}\frac{(\beta_1+v(\boldsymbol{x}))^{\alpha_1+d}}{(\alpha_1+d-1)}\left[\frac{1}{(\beta_1+w_1+v(\boldsymbol{x}))^{\alpha_1+d-1}}-\frac{1}{(\beta_1+w_2+v(\boldsymbol{x}))^{\alpha_1+d-1}}\right].
\end{align*}

Next, we provide an alternative form of consumer´s decision function. The alternative form of the decision function is useful to calculate manufacturer´s utility. For developing an alternative form of the consumer's decision function, we consider the following two results.
\begin{result}\label{r1}$A_1((v(\boldsymbol{x}),d)\ |\ \boldsymbol{m})$ is decreasing in $v(\boldsymbol{x})$ for fixed $d$.

\end{result}
\begin{proof}
We present the proof in the appendix.
\end{proof}
\begin{result}\label{r2}$\frac{a_1}{L}A_1((v(\boldsymbol{x}),d)\ |\ \boldsymbol{m})-A_2((v(\boldsymbol{x}),d)\ |\ \boldsymbol{m})$ is decreasing in $v(\boldsymbol{x})$ for fixed $d$

\end{result}
\begin{proof}
We provide the proof in the appendix.
\end{proof}
\vspace{0.5cm}\\
For fixed $d$, $A_1(v(\boldsymbol{x}),d\ | \ \boldsymbol{m})$ is decreasing in $v(\boldsymbol{x})$, there exist a point $c(d)$ such that  
\begin{align*}
    A_1((v(\boldsymbol{x}),d)\ | \ \boldsymbol{m})< \frac{L(a_3-a_2)}{a_1}, &~~~~~ \text{for } v(\boldsymbol{x})>c(d)\\
       A_1((v(\boldsymbol{x}),d)\ | \ \boldsymbol{m})> \frac{L(a_3-a_2)}{a_1}, &~~~~~ \text{for } v(\boldsymbol{x})<c(d).
\end{align*}
For fixed $d$, $\frac{a_1}{L}A_1((v(\boldsymbol{x}),d)\ |\ \boldsymbol{m})-A_2((v(\boldsymbol{x}),d)\ |\ \boldsymbol{m})$ is decreasing in $v(\boldsymbol{x})$, there exist a point $c'(d)$ such that  
\begin{align*}
 \frac{a_1}{L}A_1((v(\boldsymbol{x}),d)\ | \ \boldsymbol{m})-A_2((v(\boldsymbol{x}),d)\ | \ \boldsymbol{m})\leq (a_3-a_2), &~~~~~ \text{for } v(\boldsymbol{x})>c'(d)\\
      \frac{a_1}{L}A_1((v(\boldsymbol{x}),d)\ | \ \boldsymbol{m})-A_2((v(\boldsymbol{x}),d)\ | \ \boldsymbol{m})\leq (a_3-a_2), &~~~~~ \text{for } v(\boldsymbol{x})<c'(d).
\end{align*}
Note that $0\leq v(\boldsymbol{x})\leq nT_0$, we have $\mathcal{X}=\left\{(v(\boldsymbol{x}),d):~~v(\boldsymbol{x})>c(d)~~\&  ~~0<v(\boldsymbol{x})<nT_0 ~~\&~~ 0\leq d \leq r\right\}$, $\mathcal{Y}=\{(v(\boldsymbol{x}),d): ~~v(\boldsymbol{x})<c(d)~~\&~~v(\boldsymbol{x})>c'(d)~~\&~~0<v(\boldsymbol{x})<nT_0 ~~\& ~~0\leq d \leq r\}$ and $\mathcal{Z}=\{(v(\boldsymbol{x}),d): ~~v(\boldsymbol{x})>c'(d)~~\&~~0<v(\boldsymbol{x})<nT_0 ~~\& ~~0\leq d \leq r\}.$

Let $c_1(d)=\min\{\max\{0,c(d)\},nT_0\}$ and $c_2(d)=\min\{\max\{0,c'(d)\},c(d)\}$. Then $\mathcal{X}$, $\mathcal{Y}$ and $\mathcal{Z}$ can be written as
$\mathcal{X}=\left\{(v(\boldsymbol{x}),d):~~ ~~c_1(d)<v(\boldsymbol{x})<nT_0 ~~\&~~ 0\leq d \leq r\right\}$, $\mathcal{Y}=\{(v(\boldsymbol{x}),d): ~~c_2(d)<v(\boldsymbol{x})<c_1(d)~~\& ~~0\leq d \leq r\}$ and $\mathcal{Z}=\{(v(\boldsymbol{x}),d): ~~0<v(\boldsymbol{x})<c_2(d) ~~\& ~~0\leq d \leq r\}.$
\subsubsection{Manufacturer's decision}\label{md}
Using equations (\ref{acc}), (\ref{acc1}) and (\ref{rej}), the expected utility functions of the manufacturer corresponding to the decisions $\mathcal{A}_{wo}$, $\mathcal{A}_w$ and $\mathcal{R}$ are obtained as  
\begin{align*}
&\mathcal{U}_{\mathcal{M}}(\mathcal{A}_{Wo}\ |\ \boldsymbol{x},\theta)=b_1{\theta^{q}}\Gamma(q+1)-b_2-b_4n-b_5d-b_6\eta ,\\
&\mathcal{U}_{\mathcal{M}}(\mathcal{A}_w\ |\ \boldsymbol{x},\theta)=b_1{\theta^{q}}\Gamma(q+1)-b_2+c_w-c_s+\frac{c_s\theta}{w_2-w_1}\left[\exp(-w_1/\theta)-\exp(-w_2/\theta\right]-b_4n-b_5d-b_6\eta\\
 \text{and }~~\\ &\mathcal{U}_{\mathcal{M}}(\mathcal{R}\ | \ \boldsymbol{x},\theta)=b_3-b_4n-b_5d-b_6\eta,
\end{align*}respectively.

It is assumed that $\theta$ follows an inverse gamma distribution with parameters $(\alpha_2,\beta_2)$. Let \begin{align*}
    V(\boldsymbol{X})=\begin{cases}
    nT_0& D=0\\
     \sum_{i=1}^D X_{(i)}+(n-D)T_0 &1\leq D<r\\
   \sum_{i=1}^r X_{(i)}+(n-r)X_{(r)} &D=r.
    \end{cases}
\end{align*} Note that $v(\boldsymbol{x})$ is the observed value of $V(\boldsymbol{X})$.
Using the sets $\mathcal{X},\mathcal{Y}$ and $\mathcal{Z}$, the manufacturer's utility which is given in (\ref{manu1}), can be written as
\begin{align}\label{man}
    \psi_{\mathcal{M}}(\boldsymbol{m})=&-b_4n+b_3P(R)-b_5E[D]-b_6E[\eta]-b_2[P(A_{wo})+P(A_w)]-L_w\nonumber\\
&+b_1\Gamma(q+1)\sum_{d=0}^r\int_0^\infty\int_{c_2(d)}^{nT_0}\theta^qf_{(V(\boldsymbol{X}),D)}(y,d)p_{\mathcal{M}}(\theta)~d\theta~ dy,
\end{align}
where 
\begin{align}
   & P(A_{wo})=\sum_{d=0}^r\int_0^\infty\int_{c_1(d)}^{nT_0} f_{(V(\boldsymbol{X}),D)}(y,d)p_{\mathcal{M}}(\theta)~d\theta ~dy,\label{m1}\\
   & P(A_{w})=\sum_{d=0}^r\int_0^\infty\int_{c_2(d)}^{c_1(d)} f_{(V(\boldsymbol{X}),D)}(y,d)p_{\mathcal{M}}(\theta)~d\theta~ dy,\label{m2}\\
  & P(R)=\sum_{d=0}^r\int_0^\infty\int_{0}^{c_2(d)} f_{(V(\boldsymbol{X}),D)}(y,d)p_{\mathcal{M}}(\theta)d\theta dy,\label{m3}\\
&L_w= \sum_{d=0}^r\int_0^\infty\int_{c_2(d)}^{c_1d)}\left(c_s-\frac{c_s\theta}{w_2-w_1}\left[\exp(-w_1/\theta)-\exp(-w_2/\theta\right]-c_w\right)f_{(V(\boldsymbol{X}),D)}(y,d)p_{\mathcal{M}}(\theta)~d\theta~ dy.\label{m4}
\end{align}
and $f_{(V(\boldsymbol{X}),D)}$ is the joint distribution of $V(\boldsymbol{X})$ and $D$.
\bt\label{the1} \cite{childs2003exact}
The joint distribution of $V(\boldsymbol{X})$ and D under type-I HCS  is given by
\begin{align*}
f_{(V(\boldsymbol{X}),D)}(y,d)=\begin{cases}
\exp\left(-nT_0/\theta\right) h(y) &\text{when }d=0 \\
\binom{n}{d} \sum_{i=0}^d\binom{d}{i} \exp\left[-(n-d+i)T_0/\theta\right](-1)^i g(y-T_0(n-d+i),\frac{1}{\theta},d)&\text{when } 1
\leq d\leq r-1\\
g(y,\frac{1}{\theta},r)+r\binom{n}{r}\sum_{k=1}^r\frac{(-1)^{k}\exp\left[-(n-r+k)T_0/\theta\right]}{n-r+k}\binom{r-1}{k-1}g(y-(n-r+k)T_0,\frac{1}{\theta},r)&\text{when }d=r~~~~~~
\end{cases}
\end{align*}
where the function $g(y,\alpha,\beta)$ is the pdf of the gamma distribution with parameters $(\alpha,\beta)$ given by
\begin{align*}
    g(y,\alpha,\beta)=\begin{cases}
    \frac{\beta^{\alpha}}{\Gamma(\alpha)}y^{\alpha-1}\exp(-\beta y)& \text{for } y>0\\
    0 & \text{otherwise},
    \end{cases}
\end{align*}
and $h(y)$ is a degenerate distribution at the point $nT_0$ which is defined as \begin{align}
    h(y)=\begin{cases}
    1& \text{for } y=nT_0\\
    0 & \text{otherwise}.
    \end{cases}
\end{align}
\et

Now consider the result which is required for computation of manufacturer's utility.
\begin{result}\label{g1}
\begin{align*}
\int_0^\infty \int_{x_1}^{x_2} \theta^{-b-1} \exp(-(a+c)/\theta)g(y-c,1/\theta,p) ~dy~ d\theta=\frac{\Gamma(b)}{(a+c)^{b}}\left[I_{\eta_2}(p,b)-I_{\eta_1}(p,b)\right]
\end{align*}
where $B_\eta(p,b)$ is the incomplete beta function, $I_\eta(p,b)$ is the cdf of beta function, $\eta_{i-1}=(x_i-c)/(x_i+a)$, $i=2,3$ and $x_3=\max(x_1,c)$. $B_\eta(p,b)$ and $I_\eta(p,b)$ are defined as
\begin{align*}
    B_\eta(p,b)=\int_0^\eta x^{mp1}(1-x)^{b-1} dy
\end{align*}
and 
\begin{align*}
   I_\eta(p,b)=\frac{B_\eta(p,b)}{B(p,b)} 
\end{align*}
respectively.
\end{result}
\begin{proof}
  We present the proof in the appendix.
\end{proof}
\vspace{0.2cm}

Using Result \ref{g1} we get,
\begin{align*}
    &\int_0^\infty \int_{x_1}^{x_2}\theta^{(-(\alpha_2-l)-1)}\exp(-(\beta_2+w+(n+j-d)T_0)/\theta) g\left(x-T_0(n-d+j),\frac{1}{\theta},d\right)~ dy~d\theta\\
    &=\frac{\Gamma(\alpha_2-l)}{(\beta_2+w+(n+i-d)T)^{\alpha_2-l}}(I_{s_2}(d,\alpha_2-q)-I_{s_1}(d,\alpha_2-q))=H_{(w,l,j,d)}(s_1,s_2), say.
\end{align*}
 Also let, $h_1=\min\{\max\{x_1,(n-d+j)T_0\},x_2\}$, $h_2=\max\{\max\{x_1,(n-d+j)T_0\},x_2\}$, and $s_i=(h_i-(n-d+j)T_0)/(h_i+\beta_2+w)$, for $i=1,2$. Note that the points $s_1$ and $s_2$ varies with $w,j$ and $d$. 

\begin{result}\label{t1}
    Suppose the lifetime of the product follows $\text{Exp}(\theta)$ and $\theta$ follows inverse gamma with parameters $(\alpha_2,\beta_2)$. If   $n$ items are put on test  and test is carried out under type-I HCS, then the manufacturer's expected utility function is given by
    \small\begin{align*}
    \psi_{\mathcal{M}}(\boldsymbol{m})&=b_1\Gamma(q+1)\left(\frac{\beta_2^{\alpha_2}\Gamma(\alpha_2-q)}{\Gamma(\alpha_2)(\beta_2+nt)^{\alpha_2-q}}I(nT>c'(0))+\sum_{i=0}^d\binom{d}{i}\binom{n}{d}(-1)^i\frac{\beta_2^{\alpha_2}}{\Gamma(\alpha_2)} H_{(0,q,i,d)}{(\zeta_2,\zeta_4)}+\frac{\beta_2^{\alpha_2}}{\Gamma(\alpha_2)}H_{(0,q,r-n,r)}{(\zeta_2,\zeta_4)}\right.\\
    &\left.+\frac{\beta_2^{\alpha_2}}{\Gamma(\alpha_2)}r\binom{n}{r}\sum_{k=1}^r\frac{(-1)^{k}}{n-r+k}\binom{r-1}{k-1}H_{(0,q,k,r)}{(\zeta_2,\zeta_2)}\right)-b_2[P(A_{wo})+P(A_w)]-L_w+b_3P(R)-b_4n-b_5E[D]-b_6E[\xi],
\end{align*}
\normalsize
where $\alpha_2>q$, $\eta_1=(n-d+i)T_0$, $\eta_2=\max\{c_2(d),(n-d+i)T_0\}$, $\eta_3=\max\{c_1(d),(n-d+i)T_0\}$,   and $\eta_4=nT_0$ and 
$\zeta_i=(\eta_i-(n+i-d)T_0)/(\eta_i+\beta_2+w)$ for $i=1,\cdots,4$.
\end{result}
The expressions of $P(A_{wo})$, $P(A_w)$, $P(R)$, $L_w$, $E[D]$, $E[\xi]$ are given in the appendix.

\subsection{Weibull Distribution case}\label{wei}
Here we derive the optimal sampling plan for the Weibull distribution. Let $X$ follows Weibull distribution with cdf
\begin{align*}
    F(x|\alpha,\lambda)=1-\exp(-\lambda x^\alpha), ~~~x>0, \alpha,\lambda>0
\end{align*}
and the priors of $\alpha$ and $\lambda$ follow gamma distributions with parameters $(u_1,v_1)$ and $(c_1,d_1)$ respectively. Here $\boldsymbol{\theta}=(\alpha,\lambda)$.
The expected utility functions of the consumer corresponding to acceptance and rejection are
\begin{align*}
    \mathcal{U}^0_\mathcal{C}(\mathcal{A}\ |\ \boldsymbol{\theta})=a_1\left[1-\frac{\gamma\left(\frac{1}{\alpha},\lambda L^\alpha\right)}{\alpha\lambda^{\frac{1}{\alpha}}L}\right]+a_2
\end{align*}
and
\begin{align*}
    \mathcal{U}^0_\mathcal{C}(\mathcal{R}\ |\ \boldsymbol{\theta})=a_3,
\end{align*}
respectively, where $\gamma(s,t)$ is the lower incomplete gamma function given by $\gamma(s,t)=\int_0^t u^{s-1} e^{-u}du$. The expected acceptance utility functions of the consumer with warranty are
$\mathcal{U}^0_\mathcal{C}(\mathcal{A}\ |\ \boldsymbol{\theta})-E_{X\ | \ \boldsymbol{\theta}}(q(X))$, where 
 \begin{align}
 E_{X\ | \ \boldsymbol{\theta}}(q(X))=&\frac{c_s+c_w}{w_2-w_1}\int_{w_2}^{w_1}[1-\exp(-\lambda x^\alpha]dx\nonumber\\
    &=(c_s+c_w)\left[1-\frac{1}{\alpha\lambda^{\frac{1}{\alpha}}(w_2-w_1)}\left\{\gamma\left(\frac{1}{\alpha},w_2^*\right)-\gamma\left(\frac{1}{\alpha},w_1^*\right)\right\}\right]=Q(\boldsymbol{\theta}),~~\text{say}
\end{align}
where $w_i^*=\lambda w_i^\alpha$ for $i=1,2$.

If the inequalities (\ref{21}) and (\ref{22}) do not hold, we go for life testing under Type-I HCS. The joint posterior distribution of $\boldsymbol{\theta}=(\alpha,\lambda)$ is given by
\begin{align}\label{post1}
    p(\boldsymbol{\theta}\ |\ \boldsymbol{x})\propto \lambda^{c_1+d-1}\alpha^{u_1+d-1}\prod_{i=0}^d x_{(i)}^{\alpha-1}\exp\left[-\lambda( v(\boldsymbol{x})+d_1)\right]\exp(-v_1\alpha),
\end{align}
where $x_{(0)}=1$ and 
\begin{align}
    v(\boldsymbol{x})=\begin{cases}
    nt^\alpha& d=0\\
     \sum_{i=1}^d x_{(i)}^\alpha+(n-d)t^\alpha &1\leq d<r\\
   \sum_{i=1}^r x_{(i)}^\alpha+(n-r)x_{(r)}^\alpha &d=r.
    \end{cases}
\end{align}
The consumer's decision cannot be obtained analytically as in the case of exponential distribution. We calculate (\ref{in1}) and (\ref{in2}) by simulation, for which we draw the sample from the joint posterior distribution given in (\ref{post1}). The procedure for drawing a sample is given in Algorithm 1. For using Algorithm 1, the following properties of the joint posterior distribution given in (\ref{post1}) are needed.\\
\textbf{Properties: }
\begin{enumerate}[(I)]
    \item The conditional posterior distribution of $\lambda$ given $\alpha$ follows the gamma distribution with parameters $(c_1+d, v(\boldsymbol{x})+d_1)$.
    \item The marginal posterior distribution of $\alpha$ is log-concave if $r\geq 1$.
\end{enumerate}
The proof of the properties is given in \cite{kundu2007hybrid}. We now give the algorithm to generate $(\alpha,\lambda)$ from $p(\alpha,\lambda|\boldsymbol{x})$\\
\textbf{Algorithm 1: }
    \begin{itemize}
    \item Step 1 : The log-concave density $p(\alpha|\boldsymbol{x})$ is used to generate $\alpha^{(i)}$, using the approach outlined by Devroye \cite{devroye1984simple}.
    \item Step 2 : Generate $\lambda^{(i)}$ from gamma distribution with parameters $(c_1+d, v(\boldsymbol{x})+d_1)$ using the value $\alpha^{(i)}$.
    \item Step 3: Repeat Steps 1 and 2 for $S$ times to yield $(\alpha^{(i)},\lambda^{ (i)})$, for $i = 1,\cdots,S$.
\end{itemize}
Using the equations (\ref{acc}), (\ref{acc1}) and (\ref{rej}), the expected utility functions of the manufacturer corresponding to the decisions $\mathcal{A}_{wo}$, $\mathcal{A}_w$ and $\mathcal{R}$ are  
\begin{align*}
&\mathcal{U}_{\mathcal{M}}(\mathcal{A}_{Wo}\ |\ \boldsymbol{x},\boldsymbol{\theta})=b_1{\lambda^{-q/\alpha}}\Gamma(q/\alpha+1)-b_2-b_4n-b_5d-b_6\eta ,\\
&\mathcal{U}_{\mathcal{M}}(\mathcal{A}_w\ |\ \boldsymbol{x},\boldsymbol{\theta})=b_1{\lambda^{-q/\alpha}}\Gamma(q/\alpha+1)-b_2+c_w-(c_s+c_w)\left[1-\frac{\gamma\left(\frac{1}{\alpha},w_2^*\right)-\gamma\left(\frac{1}{\alpha},w_1^*\right)}{\alpha\lambda^{\frac{1}{\alpha}}(w_2-w_1)}\right]-b_4n-b_5d-b_6\eta\\
& \text{and }~~\\ &\mathcal{U}_{\mathcal{M}}(\mathcal{R}\ | \ \boldsymbol{x},\boldsymbol{\theta})=b_3-b_4n-b_5d-b_6\eta,
\end{align*}
respectively.
Here, analytical solution cannot be obtained as in the case of the exponential distribution. The manufacturer's expected utility can be obtained using Algorithm \ref{algo2}. 
    \RestyleAlgo{ruled}
\SetKwComment{Comment}{/* }{ */}
\begin{algorithm}[hbt!]
\setcounter{algocf}{1}
 \caption{Manufacturer's expected utility}\label{5}
Consider $\psi_0=0$\\
\For {$u=1,2,\cdots,S_1$}
{

\begin{enumerate}
    \item Generate $x^{(u)}_{(1)},\ldots,x^{(u)}_{(d^{(u)})},d^{(u)}$ and $\eta^{(u)}$ using the algorithm provided in the appendix.
     \item $\text{Generate}~~~\boldsymbol{\theta}^{(1)},\boldsymbol{\theta}^{(2)},\cdots,\boldsymbol{\theta}^{(S_2)}\sim p_\mathcal{C}(\boldsymbol{\theta}\ |\ \boldsymbol{x}^{(u)})$, where  $\boldsymbol{x}^{(u)} =(x^{(u)}_{(1)},\ldots,x^{(u)}_{(d^{(u)})},d^{(u)})$.
     \item Calculate 
     $$\mathcal{U}^0_{\mathcal{C}}(\mathcal{D}\ |\ \boldsymbol{x}^{(u)},\boldsymbol{m})=\frac{1}{S_2}\left[\mathcal{U}^0_{\mathcal{C}}\left(\mathcal{D}\ |\ \boldsymbol{\theta}^{(1)}\right)+\mathcal{U}^0_{\mathcal{C}}\left(\mathcal{D}\ |\ \boldsymbol{\theta}^{(2)}\right)+\cdots+\mathcal{U}^0_{\mathcal{C}}\left(\mathcal{D}\ |\ \boldsymbol{\theta}^{(S_2)}\right)\right],$$
where $\mathcal{D}=\{\mathcal{A},\mathcal{R}\}$  and
$$E_{\boldsymbol{\theta}\ | \ (\boldsymbol{x}^{(u)},\boldsymbol{m})}[Q(\boldsymbol{\theta})]=\frac{1}{S_2}\left[Q\left(\boldsymbol{\theta}^{(1)}\right)+Q\left(\boldsymbol{\theta}^{(2)}\right)+\cdots+Q\left(\boldsymbol{\theta}^{(S_2)}\right)\right].$$
\end{enumerate}
Consider $r_1=I_{\mathcal{X}}(x^{(u)})$, $r_2=I_{\mathcal{Y}}(x^{(u)})$ and $r_3=I_{\mathcal{Z}}(x^{(u)})$\\
\If{ (\textbf{The inequality (\ref{in1}) holds for $\boldsymbol{x}^{(u)}$})}{$r_1=1 $, $r_2=0$, $r_3=0$}
\Else{
\If{(\textbf{The inequality (\ref{in2}) holds for $\boldsymbol{x}^{(u)}$})}
{
$r_1=0$, $r_2=1$, $r_3=0$}
\Else
{
$r_1=0$, $r_2=0$, $r_3=1$
}
}
$\psi_1=(r_1+r_2)(b_1E_{X|(\boldsymbol{x}^{(u)},\boldsymbol{\theta})}[X^q]-b_2)-r_2E_{X|(\boldsymbol{x}^{(u)},\boldsymbol{\theta})}[q(X)]+r_3b_3-b_5d^{(u)}-b_6\eta_0^{(u)}$\\
$\psi_0=\psi_1+\psi_0$
}
$\psi_\mathcal{M}(\boldsymbol{m})=-b_4n+\psi_0/{S_1}$
\label{algo2}
\end{algorithm}
\section{Random decision-making sampling plan}\label{ARA}
 We now consider an approach where the manufacturer’s knowledge about the consumer’s utility, prior and actions are not required. Here, the manufacturer knows the form of the utility function but does not know the parameters of the utility function. Similarly, the form of the consumer's prior distribution is known to the manufacturer but the hyper-parameters are uncertain. The manufacturer estimates the probabilities of the consumer's action after the life test. The consumer's exact action after the life testing is not known to the manufacturer. Therefore, the manufacturer's action does not depend on the consumer's action and hence does not affect the optimal value of the sampling parameter.  

The utility functions for the consumer are discussed in equations (\ref{theta}), (\ref{theta1}) and (\ref{theta2}). The methodology is discussed for the exponential distribution and it is assumed that all the parameters of the utility function of the consumer and hyperparameter of the prior distribution follow uniform distributions independently. We consider that $\alpha_1\sim\mathcal{U}[\alpha_1^1,\alpha_1^2]$, $\beta_1\sim\mathcal{U}[\beta_1^1,\beta_1^2]$, $a_1\sim\mathcal{U}[a_1^1,a_1^2]$ and $a_2\sim \mathcal{U}[a_2^1,a_2^2]$, $a_3\sim\mathcal{U}[a_3^1,a_3^2]$, $L\sim\mathcal{U}[L^1,L^2]$ where $0<\alpha_1^1<\alpha_1^2$, $0<\beta_1^1<\beta_1^2$, $0<a_1^1<a_1^2$, $a_2^1<a_2^1$, $a_3^1<a_3^2$ and $L^1<L^2$. Here a random variable $Z\sim \mathcal{U}[a,b]$ means $Z$ follows uniform distribution whose pdf is given by
 \begin{align*}
     f(z)=\begin{cases}
         \frac{1}{b-a} & \text{for }a\leq z\leq b\\
         0& \text{otherwise }
     \end{cases}
 \end{align*}

    \RestyleAlgo{ruled}
\SetKwComment{Comment}{/* }{ */}
\begin{algorithm}[hbt!]
 \caption{Find the probabilities $P(\mathcal{A}^1_{wo}(\boldsymbol{x})\ | \ \mathcal{R}_0)$, $P(\mathcal{A}^1_{ww}(\boldsymbol{x})\ | \ \mathcal{R}_0)$ and $P(\mathcal{R}^1(\mathcal{x})\ | \ \mathcal{R}_0)$}
Consider a sample $\boldsymbol{x}=(x_{(1)},x_{(2)},\cdots,x_{(d)},d)$\\
$r_{01}=0$, $r_{02}=0$, $r_{03}=0$, $r_{11}=0$, $r_{12}=0$, $r_{13}=0$\\
\For {$i=1,2,\ldots,K$}
{

\begin{enumerate}
    \item Generate $\alpha_1^i\sim\mathcal{U}[\alpha_1^1,\alpha_1^2]$, $\beta_1^i\sim\mathcal{U}[\beta_1^1,\beta_1^2]$, $a_1^i\sim\mathcal{U}[a_1^1,a_1^2]$, $a_2^i\sim \mathcal{U}[a_2^1,a_2^2]$, $a_3^i\sim\mathcal{U}[a_3^1,a_3^2]$\\
    and $L^i\sim\mathcal{U}[L^1,L^2]$.\\
     \item Compute the expressions (\ref{theta}), (\ref{theta2}), $\mathcal{U}^0(\mathcal{A}\ | \ \boldsymbol{x},\boldsymbol{m})$ and $E_{\boldsymbol{\theta}\ | \ (\boldsymbol{x},\boldsymbol{m})}[E_{X\ | \ \boldsymbol{\theta}}[q(X)]]$
\end{enumerate}
\If{ (\textbf{The inequality (\ref{21}) holds})}{$r_{01}=r_{01}+1$}
\Else{
\If{(\textbf{The inequality (\ref{22}) holds})}
{
$r_{02}=r_{02}+1$}
\Else
{
$r_{03}=r_{03}+1$\\
\If{ (\textbf{The inequality (\ref{in1}) holds for $\boldsymbol{x}$})}{$r_{11}=r_{11}+1$}
\Else
{
\If{ (\textbf{The inequality (\ref{in2}) holds for $\boldsymbol{x}$})}{$r_{12}=r_{12}+1$}
\Else{
$r_{13}=r_{13}+1$
}
}
}
}
$P(\mathcal{A}^0_{wo})=r_{01}/K$\\
$P(\mathcal{A}^0_{w})=r_{02}/K$\\
$P(\mathcal{R}^0)=r_{03}/K$\\
$P(\mathcal{A}^1_{wo}(\boldsymbol{x})\ | \ \mathcal{R}_0)=r_{11}/r_{03}$\\
$P(\mathcal{A}^1_{w}(\boldsymbol{x})\ | \ \mathcal{R}_0)=r_{12}/r_{03}$\\
$P(\mathcal{R}^1(\boldsymbol{x})\ | \ \mathcal{R}_0)=r_{13}/r_{03}$
}
\label{algo}
\end{algorithm}
Let $\mathcal{A}^0_{wo}$, $\mathcal{A}^0_{ww}$ and $\mathcal{R}^0$ be the event that the consumer accepts the lot without warranty, with warranty and rejects the lot without inspection, respectively. Also, let $\mathcal{A}^1_{wo}(\boldsymbol{x})\ | \ \mathcal{R}_0$ be the event of acceptance without warranty by the consumer for the data $\boldsymbol{x}$ given that consumer rejects the lot without inspection, $\mathcal{A}^1_{ww}(\boldsymbol{x})\ | \ \mathcal{R}_0$ be the event of acceptance with warranty by the consumer for the data $\boldsymbol{x}$ given that consumer rejects the lot without inspection and $\mathcal{R}^1(\boldsymbol{x})\ | \ \mathcal{R}_0$ be the event of rejection by the consumer for the data $\boldsymbol{x}$ given that consumer rejects the lot without inspection.  We give an Algorithm \ref{algo} to estimate the probabilities of these events. Using estimated probabilities, compute the expected utility of the manufacturer as
    \begin{align*}
    \psi_\mathcal{M}(\boldsymbol{m})=&\int_{\theta}\int_{\boldsymbol{y}}\sum_dP(\mathcal{A}^1_{wo}(\boldsymbol{x})\ | \ \mathcal{R}_0)\left[b_1E_{X\ |\ (\boldsymbol{x},\boldsymbol{\theta})}(X^q)-b_2-b_4n-b_5d-b_6\eta\right]P(\mathcal{A}^1_{ww}(\boldsymbol{x})\ | \ \mathcal{R}_0)\left[b_1E_{X|(\boldsymbol{x},\boldsymbol{\theta})}(X^q)\right.\\
    &\left.-b_2-b_4n-b_5d-b_6\eta\right]+P(\mathcal{R}^1(\boldsymbol{x})\ | \ \mathcal{R}_0)\left[b_3-b_4n-b_5d-b_6\eta\right]f(\boldsymbol{y},d\ | \ \boldsymbol{\theta})p_{\boldsymbol{M}}(\theta)~~d\boldsymbol{x}~d\boldsymbol{y}~d\theta
    \end{align*}
    The optimal decision $\boldsymbol{m}^*=(n^*,r^*,T^*)$ is obtained by maximizing $\psi_{\mathcal{M}}(\boldsymbol{m})$, i.e,
    \begin{align*}
        \boldsymbol{m}^*=\arg\max\limits_{\boldsymbol{m}}\psi_{\mathcal{M}}(\boldsymbol{m}).
    \end{align*}

 \section{Numerical Example}\label{numerical}
 \textbf{Example 1: }The hyperparameters of the prior manufacturer are taken as $\alpha_2=1.8$ and $\beta_2=18$. The cost components and risk parameter of the manufacturer's utility function are chosen as $b_1=10$, $b_2=5$, $b_3=25$, $b_4=1$, $b_5=0.5$, $b_6=0.5$ and $q=0.8$. The warranty time points are $w_1=5$, $w_2=10$, warranty price $c_w=0.5$ and selling price of the product is $c_s=2$. The consumer's required lifetime is $L=15$.
The cost parameters of consumer utility are $a_1=10$, $a_2=5$, $a_3=9$ and the hyperparameters of the consumer are $\alpha_1=2$ and $\beta_1=3$. The optimal design, the manufacturer's utility, the probability of acceptance, rejection after life testing, the expected number of failure and the expected time duration of the test and given in Table \ref{tab 22}. 
\begin{table}[ht!]
\caption{Optimal design for the exponential distribution.}
    \begin{center}
        \begin{tabular}{|cccccc|}
       \hline
     $\boldsymbol{m}^*=(n^*,r^*,T_0^*)$& $\psi_{\mathcal{M}}(\boldsymbol{m}^*)$&$[P(A_{wo}),P(A_{ww}),P(R)]$ & $E[D]$&$E[\xi]$ & $L_w$ \\
         \hline
       (5, 2, 5.75)&70.81&[0.18, 0.24, 0.58]&1.40&3.95&0.09\\
       \hline
    \end{tabular}
    \end{center}
    \label{tab 22}
\end{table}
Next, we consider the effect of parameters on optimum solution. The effect of the parameters $a_3$, $b_1$, $b_3$, $(\alpha_2,\beta_2)$ and $b_6$ are given in the Tables \ref{tab 1}, \ref{tab 2}, \ref{tab 4}, \ref{tab 5} and \ref{tab 6}, respectively.
\begin{table}[ht!]
\caption{Optimal design for different values of $a_3$ for the exponential distribution.}
    \begin{center}
        \begin{tabular}{|ccccccc|}
       \hline
     $a_3$ &$\boldsymbol{m}^*=(n^*,r^*,T_0^*)$& $\psi_{\mathcal{M}}(\boldsymbol{m}^*)$&$[P(A_{wo}),P(A_{ww}),P(R)]$ & $E[D]$&$E[\xi]$ & $L_w$ \\
         \hline
         8&(5, 2, 10.63)&57.69&[0.10, 0.13, 0.77]&1.69&5.46&0.02\\
       9&(5, 2, 5.75)&70.81&[0.18, 0.24, 0.58]&1.40&3.95&0.09\\
       10&(4, 2, 4.03)&81.30&[0.32, 0.32, 0.36]&1.05&3.35&0.19\\
       11&(1, 1, 1.71) &89.38&[0, 0.85, 0.15]&0.15&1.57&0.50\\
       12&(0, 0, 0)&92.03&[0, 1, 0]&0&0&0.92\\
       \hline
    \end{tabular}
    
    \end{center}

    \label{tab 1}
\end{table}

It is observed in Table \ref{tab 1} that when $a_3$ increases for fixed values of other parameters, the probability of acceptance without warranty increases. This is due to the fact that the lot is accepted without warranty, which is the first priority to the consumer because of the low price of the product. When $a_3=12$, the optimal design is $(0,0,0)$ and $P(A_{ww})=1$. This means that the lot is accepted with a warranty and without life-testing. Also, it is observed that when $a_3$ increases, the sample size and time point decrease. This happens because as $a_3$ increases, the value of  $E_{p_{\mathcal{C}}}(\theta)[\mathcal{U}_{\mathcal{C}}(\mathcal{A}\ | \ \theta)]$ gets closer to $a_3$. 
\begin{table}[ht!]
 \caption{Optimal design for different values of $b_1$ for the exponential distribution.}
    \begin{center}
        \begin{tabular}{|ccccccc|}
       \hline
     $b_1$&$\boldsymbol{m}^*=(n^*,r^*,T_0^*)$& $\psi_{\mathcal{M}}(\boldsymbol{m}^*)$&$[P(A_{wo}),P(A_{ww}),P(R)]$ & $E[D]$&$E[\xi]$ & $L_w$\\
         \hline
         3&(5, 1, 3.30)&29.31&[0, 0.31, 0.69]&0.52&3.12&0.21\\
       5&(2, 1, 6.55)&38.59&[0, 0.37, 0.63]&0.63&3.99&0.11\\
       10&(5, 2, 5.75)&70.81&[0.18, 0.24, 0.58]&1.40&3.95&0
       09\\
       15&(8, 5, 10.21) &106.85&[0.34, 0.13, 0.53]&3.75&7.89&0.07\\
       20&(10, 7, 11.76)&144.26&[0.36, 0.14, 0.50]&5.28&9.47&0.07\\
       \hline
    \end{tabular}
    
    \end{center}

    \label{tab 2}
\end{table}

We provide effect of $b_1$ for fixed values of other parameters in Table \ref{tab 2}, it is observed that when $b_1$ increases, the probability of rejection of the lot decreases. This is due to the fact that if $b_1$ increases, the manufacturer's utility of acceptance increases. Also, we see that sample size increases with $b_1$. Similarly, in Table \ref{tab 5},  when $b_3$ increases, the probability of rejection of the lot increases. This is due to the fact that if $b_3$ increases, the manufacturer's utility of rejection increases. Also, it is seen that the sample size decreases with $b_3$.

\begin{table}[ht!]
\caption{Optimal design for different values of $b_3$ for the exponential distribution.}
    \begin{center}
        \begin{tabular}
    {|ccccccc|}
       \hline  $b_3$& $\boldsymbol{m}^*=(n^*,r^*,T_0^*)$& $\psi_{\mathcal{M}}(\boldsymbol{m}^*)$&$[P(A_{wo}),P(A_{ww}),P(R)]$ & $E[D]$&$E[\xi]$ & $L_w$\\
         \hline
         15&(8, 4, 7.46)&65.22&[0.32, 0.14, 0.54]&2.94&5.65&0.07\\
      35&(5, 2, 5.75)&76.62&[0.18, 0.24, 0.58] &1.40&3.95&0.09\\
     55&(4, 2, 7.67)&88.37&[0.17, 0.24, 0.69] &1.42&5.21&0.09\\
      110 &(2, 1, 11.50)&123.10&[0.23, 0, 0.77]&0.77&5.43&0\\
       \hline
    \end{tabular}
    
    \end{center}

    \label{tab 5}
\end{table}
\begin{table}[ht!]
\caption{Optimal design for different values of $(\alpha_2,\beta_2)$ for the exponential distribution.}
    \begin{center}
        \begin{tabular}
    {|ccccccc|}
       \hline  $(\alpha_2,\beta_2)$& $\boldsymbol{m}^*=(n^*,r^*,T_0^*)$& $\psi_{\mathcal{M}}(\boldsymbol{m}^*)$&$[P(A_{wo}),P(A_{ww}),P(R)]$ & $E[D]$&$E[\xi]$ & $L_w$ \\
         \hline
         (2.8, 18)&(2, 1, 6.55)&33.21&[0, 0.22, 0.78]&0.78&3.13&0.14\\
         (2.8, 28)&(5, 2, 5.75)&48.42&[0.14, 0.24, 0.62]&1.48&3.88&0.11\\
       (1.8, 28)&(8, 5, 10.13)&112.96&[0.51, 0.15, 0.34]&3.10&8.93&0.06\\
      (18,180)&(2, 1, 6.55)&30.40&[0, 0.28, 0.72]&0.72&3.69&0.20\\
       \hline
    \end{tabular}
    
    \end{center}

    \label{tab 4}
\end{table}
\begin{table}[ht!]
\caption{Optimal design for different values of $b_6$ for the exponential distribution.}
    \begin{center}
        \begin{tabular}
    {|ccccccc|}
       \hline  $b_6$& $\boldsymbol{m}^*=(n^*,r^*,T_0^*)$& $\psi_{\mathcal{M}}(\boldsymbol{m}^*)$&$[P(A_{wo}),P(A_{ww}),P(R)]$ & $E[D]$&$E[\xi]$ & $L_w$\\
         \hline
        0& (5, 5, 33.29)&76.41&[0.33, 0.13, 0.54]&4.24&23.06&0.06\\
      0.1&(6, 5, 17.37)&74.36&[0.33, 0.13, 0.54]&3.91&13.18&0.07\\
     1&(6, 2, 4.60)&68.95&[0.19, 0.24, 0.57]&1.39&3.19&0.09\\
     3&(8, 2, 3.29)&63.67&[0.20, 0.23, 0.57]&1.37&2.30&0.09\\
       \hline
    \end{tabular}
    
    \end{center}

    \label{tab 6}
\end{table}


\textbf{Example 2: }  Here we consider RDSP and assume that the lifetime follows the exponential distribution with pdf $f_\theta(t)=\frac{1}{\theta}e^{-\frac{t}{\theta}}$, $t>0$, $\theta>0$. The parameters of the consumer's utility and prior are random. We take $a_1\sim\mathcal{U}(10,20)$, $a_2\sim\mathcal{U}(1,7)$, $a_3\sim\mathcal{U}(3,12)$ $\alpha_1\sim\mathcal{U}(1,8)$, $\beta_1\sim\mathcal{U}(1.5,3.5)$ and $L\sim\mathcal{U}(12,18)$. The cost parameters of the manufacturer's utility are the same as in Example 1. The optimal design is given in Table \ref{tab 81} and the effect of the parameter $b_1$ is given in Table \ref{tab 8}.

\begin{table}[ht!]
\caption{Optimal design for different values of $b_1$ for the exponential distribution under RDSP approach.}
    \begin{center}
        \begin{tabular}
    {|cccccc|}
       \hline  $\boldsymbol{m}^*=(n^*,r^*,T_0^*)$& $\psi_{\mathcal{M}}(\boldsymbol{m}^*)$&$[P(A_{wo}),P(A_{ww}),P(R)]$ & $E[D]$&$E[\xi]$ & $L_w$\\
         \hline
      (3, 3, 4.73)&85.08&[0.50, 0.37, 0.13]&1.03&4.59&1.14\\
  
       \hline
    \end{tabular}
    
    \end{center}

    \label{tab 81}
\end{table}
\begin{table}[ht!]
\caption{Optimal design for different values of $b_1$ for the exponential distribution under RDSP approach.}
    \begin{center}
        \begin{tabular}
    {|ccccccc|}
       \hline  $b_1$& $\boldsymbol{m}^*=(n^*,r^*,T_0^*)$& $\psi_{\mathcal{M}}(\boldsymbol{m}^*)$&$[P(A_{wo}),P(A_{ww}),P(R)]$ & $E[D]$&$E[\xi]$ & $L_w$\\
         \hline
        5& (2, 1, 5.45)&39.37&[0.46, 0.27, 0.27]&0.57&3.56&0.83\\
      10&(3, 3, 4.73)&85.08&[0.50, 0.37, 0.13]&1.03&4.59&1.14\\
     15&(5,4,4.93)&132.38&[0.68, 0.23, 0.09]&1.62&4.40&0.70\\
     20&(5, 5, 4.94)&180.13&[0.69, 0.22, 0.09]&1.74&4.90&0.68\\
       \hline
    \end{tabular}
    
    \end{center}

    \label{tab 8}
\end{table}

In Table \ref{tab 8}, it is observed that when $b_1$ increases, the probability of rejection of the lot decreases. This is due to the fact that if $b_1$ increases, the manufacturer's utility of acceptance increases, which is similar to the non-random case for exponential distribution given in subsection \ref{exp}. Also, it is seen that sample size increases with $b_1$.

\section{Application}\label{real}
The proposed methodology of determining optimum RASP is illustrated by using data on  failure times (in hours) of the air-conditioning system of plane "7913” taken from Proschan \cite{proschan1963theoretical}. Mondal and Kundu \cite{mondal2020bayesian} analyzed this data dividing by 100 and fitted a Weibull distribution whose pdf is given by
\begin{align*}
    f(x\ | \ \alpha, \lambda)=\alpha\lambda x^{\alpha-1} \exp(-\lambda x^\alpha).
\end{align*}
The ML estimates of $\alpha$ and $\lambda$ obtained in Mondal and Kundu \cite{mondal2020bayesian} for the complete data are
$\hat{\alpha}=1.123$ and $\hat{\lambda}=1.126$, respectively. Suppose that an aircraft component manufacturer (manufacturer) is negotiating with an aircraft manufacturer (consumer) for the sale of air conditioning systems. Both manufacturer and consumer agree that the lifetime follows Weibull distribution. We assume that the manufacturer considered gamma priors for $\alpha$ and $\lambda$ with respective means 1.123 and 1.126. The consumer also considered gamma priors for $\alpha$ and $\lambda$. The mean values of $\alpha$ and $\lambda$ considered by the consumer are 1.123 and  2.252, respectively. The hyper-parameters of gamma priors of parameters $\alpha$ and $\lambda$ for the consumer are $u_1=11.23$, $v_1=10$ and $c_1=22.52$, $d_1=10$, respectively and for the manufacturer $u_2=112.6$ $v_2=100$ and $c_2=112.3$, $d_2=100$ and, respectively. The consumer believes that the mean lifetime of the product is $0.465$ hour and the manufacturer believes that the mean lifetime of the product is $0.862$ hour.

Next, we consider the cost components and risk parameter of the manufacturer, which are taken as $b_1=3000$, $b_2=1000$, $b_3=250$, $b_4=30$, $b_5=10$, $b_6=5$ and $q=1$. The warranty time points are taken as $w_1=0.2\text{ hour}$ and $w_2=0.3 \text{ hour}$, warranty price $c_w=200$ and selling price of the product is $c_s=2000$. The consumer's requirement of a minimum lifetime is $L=0.5$ hour. The cost parameters of consumer utility are $a_1=12000$, $a_2=5000$ and $a_3=8250$.  For this value, the optimum RASP is $(n,r,T_0)=(10,5,0.481)$. The probability of acceptance of the lot without a warranty after life testing is 0.35; the probability of acceptance with a warranty after life testing is 0.24 and rejection of the lot after life testing is 0.41.

Next, we have to carry out a life test under the optimum RASP. For illustration,
we generate Type-I hybrid censored data based on the optimum life testing plan $\boldsymbol{m}=(10,5,0.481)$ from the Weibull distribution with $\alpha=1.123$ and $\lambda=1.126$. Then we calculate $e_1=\mathcal{U}_{\mathcal{C}}^0(\mathcal{A}\ | \ \boldsymbol{x}, \boldsymbol{m})-a_3$ and $e_2=\mathcal{U}_{\mathcal{C}}^0(\mathcal{A}\ | \ \boldsymbol{x}, \boldsymbol{m})-E_{\boldsymbol{\theta}\ | \ (\boldsymbol{x},\boldsymbol{m})}[E_{X\ | \ \boldsymbol{\theta}}(q(X)]+c_w-a_3$. If $e_1<0$, the consumer accepts the lot without warranty; if $e_1>0$ and $e_2<0$, the consumer accepts the lot with warranty; and if $e_1,e_2>0$, the consumer rejects the lot. Some data sets and corresponding decisions are given in Table \ref{dectable} for illustration purposes.
\begin{table}[hbt!] 
    \centering
     \caption{The data sets and consumer's acceptance utility and his/her decision}
   \resizebox{\textwidth}{!}{ \begin{tabular}
{|c|cccccccc|}
    \hline
       $i$& $x_1$ &$x_2$ &$x_3$&$x_4$&$x_5$ & $e_1$ &$e_2$& Decision \\
        \hline
     $1$& $0.243$  &$0.354$&$0.457$&-&-&$-512.83$&$-787.44$&accept the lot without warranty\\
     $2$&$0.020$ &$0.155$ &$0.272$ &$0.423$&-&$645.61$&$107.01$&reject the lot\\
     $3$&  $0.150$& $0.220$& $0.250$ &$0.465$&-&$42.01$& $-356.93$&accpet the lot with warranty\\
      $4$& $0.103$& $0.151$& $0.230$& $0.405$&$0.420$&$441.06$&$-45.22$&accept the lot with warranty\\
       $5$& $0.038$&-&-&-&-&$-173.01$&$-534.39$&accept the lot without warranty\\
        \hline
    \end{tabular}
   }
    \label{dectable}
\end{table}
 
 \section*{Conclusion}\label{conclu}
 In this work, we have considered determination of optimum Bayesian RASP with optional warranty under hybrid censoring. The work can be extended to other censoring schemes. We have considered exponential and Weibull distributions for illustration. The proposed methodology can be extended to other lifetime distributions with an appropriate choice of the prior distributions.
Here we have considered that only one consumer negotiated with one manufacturer. However, in many situations, the consumer negotiates with more than one manufacturer for taking a decision to buy a product. This will be considered in future studies. 
 \bibliographystyle{plain}
\bibliography{citiation}

\appendix
\section*{Appendix}
\setcounter{lemma}{0}
    \renewcommand{\thelemma}{\Alph{section}\arabic{lemma}}
To prove Results \ref{r1} and \ref{r2}, we need to prove the following lemma.
\begin{lemma}\label{th1}
Suppose $n$ identical units are put on a life test under the Type-I HCS. The lifetimes of the units are iid with cdf $F_{{\theta}}(x)=1-\exp(-x/\theta)$. Let $p_{\mathcal{C}}(\theta)$ be the prior of the consumer. Then for any positive decreasing function $h(\theta)$, the posterior expectation of $h(\theta)$ is decreasing in $v(\boldsymbol{x})$ for fixed $d$. 
\end{lemma}
\begin{proof}
From equation (\ref{e1}), the likelihood function is given by
\begin{align*}
    L(\theta\ |\ \boldsymbol{x},\boldsymbol{m})\propto \theta^{-d}\exp\left(-\frac{v(\boldsymbol{x})}{\theta}\right).
\end{align*}
The posterior expectation of $h(\theta)$, denoted by $A((v(\boldsymbol{x}),d)\ |\ m)$, is given by
\begin{align*}
    A((v(\boldsymbol{x}),d)\ |\ \boldsymbol{m})=\frac{\int_0^\infty h(\theta)\theta^{-d}\exp\left(-\frac{v(\boldsymbol{x})}{\theta}\right)p_{\mathcal{C}}(\theta)~d\theta}{\int_0^\infty\theta^{-d}\exp\left(-\frac{v(\boldsymbol{x})}{\theta}\right)p_{\mathcal{C}}(\theta)~d\theta}.
\end{align*}
Consider two points $v(\boldsymbol{x})_1$ and $v(\boldsymbol{x})_2$ such that $v(\boldsymbol{x})_1<v(\boldsymbol{x})_2$, it suffices to show that $A((v(\boldsymbol{x})_1,d)\ |\ \boldsymbol{m})<(>)A((v(\boldsymbol{x})_2,d)\ |\ \boldsymbol{m})$ when $h(\theta)$ is increasing (decreasing) function. Let $f_1(\theta)=\theta^{-d}\exp\left(-\frac{v(\boldsymbol{x})_1}{\theta}\right)$, $f_2(\theta)=\theta^{-d}\exp\left(-\frac{v(\boldsymbol{x})_2}{\theta}\right)$, $g_1(\theta)=h(\theta)p_{\mathcal{C}}(\theta)$ and $g_2(\theta)=p_{\mathcal{C}}(\theta)$. Thus
\begin{align*}
    A((v(\boldsymbol{x})_1,d)\ |\ \boldsymbol{m})=\frac{\int_0^\infty f_1(\theta)g_1(\theta)~d\theta}{\int_0^\infty f_1(\theta)g_2(\theta)~d\theta}
\end{align*}
and
\begin{align*}
    A((v(\boldsymbol{x})_2,d)\ |\ \boldsymbol{m})=\frac{\int_0^\infty f_2(\theta)g_1(\theta)~d\theta}{\int_0^\infty f_2(\theta)g_2(\theta)~d\theta}
\end{align*}
We assume all integrals are finite. Clearly, $f_2(\theta)$ and $g_2(\theta)$ are non negative functions of $\theta$. Note that $$\frac{f_1(\theta)}{f_2(\theta)}=\exp\left[\frac{(v(\boldsymbol{x})_2-v(\boldsymbol{x})_1)}{\theta}\right]$$ is decreasing in $\theta$ when $v(\boldsymbol{x})_1<v(\boldsymbol{x})_2$. and $g_1(\theta)/g_2(\theta)=h(\theta)$ which is decreasing function in $\theta$ for $\theta>0$. By Theorem 2 in Wijsman \cite{wijsman1985useful}, we get $A((v(\boldsymbol{x})_1,d)\ |\ \boldsymbol{m})>A((v(\boldsymbol{x})_2,d)\ |\ \boldsymbol{m})$ when $v(\boldsymbol{x})_1<v(\boldsymbol{x})_2$. Therefore, $A((v(\boldsymbol{x}),d)\ | \ \boldsymbol{m})$ is decreasing in $v(\boldsymbol{x})$ for fixed $d$.
\end{proof}
\vspace{0.5cm}\\
\textbf{Proof of Result \ref{r1}:}
\begin{proof}
    Note that
\begin{align*}
    A_1((v(\boldsymbol{x}),d)\ | \ \boldsymbol{m})=&\int_0^\infty \left(\int_0^L [1-\exp(-x/\theta)]dx\right) p_{\mathcal{C}}(\theta\ | \ \boldsymbol{x},\boldsymbol{m}) d\theta\\
    &=\int_0^\infty h(\theta) p_{\mathcal{C}}(\theta\ | \boldsymbol{x},\boldsymbol{m}) d\theta,
\end{align*}
where $h(\theta)=\int_0^L [1-\exp(-x/\theta)]dx.$ Note that $h(\theta)$ is decreasing in $\theta$. By Lemma \ref{th1}, it follows that $A_1((v(\boldsymbol{x},d)\ | \ \boldsymbol{m})$ is decreasing in $v(\boldsymbol{x})$ for fixed $d$.
\end{proof}
\vspace{0.5cm}\\
\textbf{Proof of Result \ref{r2}:}
\begin{proof}
    We have
\begin{align*}
  &  \frac{a_1}{L}A_1((v(\boldsymbol{x}),d)\ | \ \boldsymbol{m})-A_2((v(\boldsymbol{x}),d)\ | \ \boldsymbol{m})\\
    &=\int_0^\infty \left[\int_0^L \frac{a_1}{L}[1-\exp(-x/\theta)]dx -\frac{c_s}{w_2-w_1}\int_{w_1}^{w_2} [1-\exp(-x/\theta)]dx\right] p_{\mathcal{C}}(\theta\ | \ \boldsymbol{x},\boldsymbol{m}) d\theta \\   &=\int_0^\infty h(\theta) p_{\mathcal{C}}(\theta\ | \boldsymbol{x},\boldsymbol{m}) d\theta,
\end{align*}
where $h(\theta)=\int_0^L \frac{a_1}{L}[1-\exp(-x/\theta)]dx -\frac{c_s}{w_2-w_1}\int_{w_1}^{w_2} [1-\exp(-x/\theta)]dx.$ Note that $h(\theta)$ is decreasing in $\theta$. By Lemma \ref{th1}, it follows that $A_1((v(\boldsymbol{x},d)\ | \ \boldsymbol{m})$ is decreasing in $v(\boldsymbol{x})$ for fixed $d$.
\end{proof}
\vspace{0.5cm}\\
\textbf{Proof of Result \ref{g1}:}
\begin{proof}
  \allowdisplaybreaks\begin{align*}
    &\int_0^\infty \int_{x_1}^{x_2} \lambda^{-b-1} \exp(-(a+c)/\lambda)g(y-c,1/\lambda,p) dy d\lambda\\
    & =\int_0^\infty \int_{x_3}^{x_2}  \lambda^{-b} \exp(-(a+c)/\lambda) \frac{\lambda^{-p}}{\Gamma(p)}(y-c)^{p-1} \exp(-(y-c)/\lambda)dy d\lambda~~~~~~~~[x_3=\min(c,x_1)]\\
    &=\frac{1}{\Gamma(p)}\int_0^\infty \int_{x_3}^{x_2}  \lambda^{-(b+p)-1} \exp(-(a+y)/\lambda) (y-c)^{p-1}dy d\lambda\\
    &=\frac{1}{\Gamma(p)}\int_{x_3}^{x_2} (y-c)^{p-1}\frac{\Gamma(b+p)}{(a+y)^{b+p}}dy\\
    &=\frac{\Gamma(b+p)}{\Gamma(p)}\int_0^{x-c} \frac{z^{p-1}}{(a+z+c)^{b+p}}dz\\
    &=\frac{\Gamma(b+p)}{\Gamma(p)(a+c)^{b+p}}\int_{x_3-c}^{x_2-c} \frac{z^{p-1}}{\left(1+\frac{z}{a+c}\right)^{b+p}}dz\\
    &=\frac{\Gamma(b+p)}{\Gamma(p)(a+c)^{b}}\int_{(x_3-c)/(a+c)}^{(x_2-c)/(a+c)} \frac{v^{p-1}}{\left(1+v\right)^{b+p}}dv\\
    &=\frac{\Gamma(b)}{(a+c)^{b}}\frac{B_{\eta_2}(p,b)-B_{\eta_1}(p,b)}{B(p,b)}\\
    &=\frac{\Gamma(b)}{(a+c)^{b}}\left[I_{\eta_2}(p,b)-I_{\eta_1}(p,b)\right]
\end{align*}  
\end{proof}
\vspace{0.5cm}\\
\textbf{Expressions for $\boldsymbol{P(A_{wo}), P(A_w), P(R), L_w, E[D], E[\eta]}$: }The expressions $E[D]$ and $E[\eta]$ are given in Liang and Yang \cite{liang2013optimal}. For the other expressions of the result, we put the value of $P(v(d)=x,D=d)$ which is given in Theorem \ref{the1}.
\begin{align*}
    P(A_{wo})=\frac{\beta_2^{\alpha_2}}{(\beta_2+nT_0)^{\alpha_2}}I(nT_0>c(0))+\frac{\beta_2^{\alpha_2}}{\Gamma(\alpha_2)}\left[\sum_{d=1}^r\sum_{i=0}^d\binom{d}{i}\binom{n}{d}(-1)^iH_{(0,0,i,d)}{(\zeta_3,\zeta_4)}+H_{(0,0,r-n,r)}{(\zeta_3,\zeta_4)}\right.\\
 \left.+r\binom{n}{r}\sum_{k=1}^r\frac{(-1)^{k}}{n-r+k}\binom{r-1}{k-1}H_{(0,0,k,r)}{(\zeta_3,\zeta_4)}\right],
\end{align*}
\begin{align*}
    P(A_{w})=\frac{\beta_2^{\alpha_2}}{(\beta_2+nT_0)^{\alpha_2}}I(c'(0)<nT_0\leq c(0))+\frac{\beta_2^{\alpha_2}}{\Gamma(\alpha_2)}\left[\sum_{d=1}^r\sum_{i=0}^d\binom{d}{i}\binom{n}{d}(-1)^iH_{(0,0,i,d)}{(\zeta_2,\zeta_3)}+H_{(0,0,r-n,r)}{(\zeta_2,\zeta_3)}\right.\\
    \left.+r\binom{n}{r}\sum_{k=1}^r\frac{(-1)^{k}}{n-r+k}\binom{r-1}{k-1}H_{(0,0,k,r)}{(\zeta_1,\zeta_2)}\right],
\end{align*}
\begin{align*}
    P(R)=\frac{\beta_2^{\alpha_2}}{(\beta_2+nT_0)^{\alpha_2}}I(nT_0\leq c(0))+\frac{\beta_2^{\alpha_2}}{\Gamma(\alpha_2)}\left[\sum_{d=1}^r\sum_{i=0}^d\binom{d}{i}\binom{n}{d}(-1)^iH_{(0,0,i,d)}{(\zeta_1,\zeta_2)}+H_{(0,0,r-n,r)}{(\zeta_1,\zeta_2)}\right.\\
    \left.+r\binom{n}{r}\sum_{k=1}^r\frac{(-1)^{k}}{n-r+k}\binom{r-1}{k-1}H_{(0,0,k,r)}{(\zeta_3,\zeta_2)}\right],
\end{align*}
\begin{align*}
    L_w=\frac{c_s}{w_2-w_1}\left(\sum_{j=1}^2(-1)^{j-1}\left[\frac{\beta_2^{\alpha_2}}{\alpha_2(\beta_2+w_j+nt)^{\alpha_2-1}}I(c'(0)<nT<c(0))+\sum_{i=0}^d\binom{d}{i}\binom{n}{d}(-1)^i\frac{\beta_2^{\alpha_2}}{\Gamma(\alpha_2)}\sum_{d=1}^rH_{(w_1,1,i,d)}{(\zeta_2,\zeta_3)}\right.\right.\\
    \left.\left.+\frac{\beta_2^{\alpha_2}}{\Gamma(\alpha_2)}H_{(w_1,1,r-n,r)}{(\zeta_2,\zeta_3)}+\frac{\beta_2^{\alpha_2}}{\Gamma(\alpha_2)}r\binom{n}{r}\sum_{k=1}^r\frac{(-1)^{k}}{n-r+k}\binom{r-1}{k-1}H_{(w_1,1,k,r)}{(\zeta_2,\zeta_3}\right]\right)+c_sP(A_w)-c_wP(A_w),
\end{align*}
 \begin{align*}
 E[D]=&{\beta_2^{\alpha_2}}\left[\sum_{j=0}^{r-1}j\binom{n}{j}\sum_{i=0}^j\binom{j}{i}(-1)^{j-i}\frac{1}{(T_0(n-i)+\beta_2)^{\alpha_2}}+r\sum_{j=r}^n\binom{n}{j}\sum_{i=0}^j\binom{j}{i}(-1)^{j-i}\frac{1}{(T_0(n-i)+\beta_2)^{\alpha_2}}\right]
\end{align*}
and
\begin{align*}
   E[\eta] =&\left[T_0\left(1-\sum_{j=r}^n\binom{n}{j}\sum_{i=0}^j\binom{j}{i}(-1)^{j-i}\frac{\beta_2^{\alpha_2}}{(T_0(n-i)+\beta_2)^{\alpha_2}}\right)\right.+r\binom{n}{r}\sum_{i=0}^{r-1}(-1)^{r-1-i}\binom{r-1}{j}\left[\frac{\beta_2}{(\alpha_2-1)}\right.\\
    &-\left.\left.\frac{\alpha_2T_0(n-i)+\beta_1}{(\alpha_2-1)}\frac{\beta_2^{\alpha_2}}{(T_0(n-i)+\beta_2)^{\alpha_2}}\right]\frac{1}{(n-i)^2}\right].\\
\end{align*}
\textbf{Generate sample for HCS: } 
The algorithm generates a pseudo-random type-I hybrid censored sample with $n$ units and $r$ failures and truncation time $T_0$ \cite{william1998statistical}. Define $U_{(0)}^{(u)}=0$, start with $i=1$, and generate the sequence as follows:
\begin{enumerate}
\item Generate $\boldsymbol{\theta}_i^{(u)}$ from $p_\mathcal{M}(\boldsymbol{\theta})$.
\item A pseudo-random observation $U_i^{(u)}$ is generated form a U(0,1)  . Compute $U_{(i)}^{(u)}=1-[1-U_{(i-1)}^{(u)}]\times(1-U_i^{(u)})^{1/(n-i+1)}$
 and $x_{(i)}^{(u)}=F^{-1}[U_{(i)}^{(u)};\boldsymbol{\theta}_i^{(u)}]$
\item If $i> r$ or $x_{(i)}^{(u)}> T_0$, stop; and take $d^{(u)}=(i-1)$ and $\eta_0^{(u)}=x_{(i-1)}^{(u)}$ the sample data is represented by $\boldsymbol{x}^{(u)}=(x_{(1)}^{(u)},\cdots,x_{(i-1)}^{(u)},d^{(u)})$.
\item If $i\leq r$ and $x_{i}^{(u)}\leq T_0$, increase the value of $i$ by 1 and return to Step 1.
\end{enumerate}

\end{document}